\documentclass[onecolumn,superscriptaddress,notitlepage]{revtex4}
\usepackage{amsmath}
\usepackage{graphicx}
\usepackage{bm}
\usepackage{amsbsy}
\usepackage{amsfonts}
\usepackage{amsthm}

\newtheorem{theorem}{Theorem}
\newtheorem{lemma}{Lemma}
\newtheorem{definition}{Definition}
\newtheorem{corollary}{Corollary}

\newenvironment{proofof}[1]{\begin{trivlist}\item[]{\flushleft\it
Proof of~#1. }}
{\qed\end{trivlist}}

\newcommand{\diff}[2]{\frac{\partial #1}{\partial #2}}

\newcommand{\gap}{\gamma}
\newcommand{\Ffunc}{f}
\newcommand{\Jfunc}{J'}

\newcommand{\sumOneJump}{\int \left(e^{-i(\int_{s_1}^1E_{\nu_1}(\xi)\mathrm{d}\xi+\int_{0}^{s_1}E_{\gst}(\xi)\mathrm{d}\xi)T}\beta_{\nu_1,\gst}(s_1)\right)\ket{\nu_1(1)}\mathrm{d}\left(\Jfunc_1 \right)}
\newcommand{\sumTwoJump}{\int\left(e^{-i(\int_{s_2}^1E_{\nu_2}(\xi)\mathrm{d}\xi+\int_{s_1}^{s_2}E_{\nu_1}(\xi)\mathrm{d}\xi+\int_{0}^{s_1}E_{\gst}(\xi)\mathrm{d}\xi)T}\beta_{\nu_2,\nu_1}(s_2)\beta_{\nu_1,\gst}(s_1)\right)\ket{\nu_2(1)}\mathrm{d}\left(\Jfunc_2 \right)}
\newcommand{\sumOneJumpxpr}{C_1}
\newcommand{\sumTwoJumpxpr}{C_2}

\newcommand{\gfunc}[2]{\kappa_{#2}(#1)}
\newcommand{\gfuncp}[2]{\gap_{#2}(#1)}

\newcommand{\bra}[1]{\langle #1|}
\newcommand{\ket}[1]{|#1\rangle}
\newcommand{\braket}[2]{\langle #1|#2\rangle}
\newcommand{\ketbra}[2]{|#1\rangle\!\langle #2|}

\newcommand{\Th}[1]{$#1^{\text{th}}$}
\newcommand{\eqrefb}[1]{Eq.~\ref{#1}}
\newcommand{\gst}{G}
\newcommand{\nn}{\nonumber\\}

\begin{document}
\bibliographystyle{unsrt}

\title{Improved Error Bounds for the Adiabatic Approximation}
\author{Donny Cheung}
\affiliation{Institute for Quantum Information Science, University of Calgary, Alberta, Canada}

\author{Peter H{\o}yer}
\affiliation{Institute for Quantum Information Science, University of Calgary, Alberta, Canada}

\author{Nathan Wiebe}
\affiliation{Institute for Quantum Information Science, University of Calgary, Alberta, Canada}

\begin{abstract}
Since the discovery of adiabatic quantum computing, a need has arisen for rigorously proven bounds for
the error in the adiabatic approximation.  We present in this paper, a rigorous and elementary derivation of
 upper and lower bounds on the error incurred
from using the adiabatic approximation for quantum systems.  Our bounds are often asymptotically tight in the limit of slow evolution for fixed Hamiltonians, and are used to provide sufficient
 conditions for the application of the adiabatic approximation.  We show that our sufficiency criteria exclude the
Marzlin--Sanders counterexample from the class of Hamiltonians that obey the adiabatic approximation. Finally, we demonstrate the existence
of classes of Hamiltonians that resemble the Marzlin--Sanders counterexample Hamiltonian, but also obey the adiabatic approximation.
\end{abstract}

\maketitle
\section{Introduction}
The adiabatic approximation is a technique of central importance to the study of time-dependent quantum
systems.  Its importance is derived from the fact that adiabatic evolution is one of the few known limiting cases where the dynamics that is generated
by a time-dependent Hamiltonian is well understood.  The approximation states that if the Hamiltonian that generates a quantum system's
evolution changes sufficiently slowly, then a quantum state that is prepared in an instantaneous eigenstate
remains in an instantaneous eigenstate throughout the evolution.  The adiabatic approximation
 has been used to control the evolution of quantum systems~\cite{oreg:stirap,kuklinski:stirap},
analyze particle collisions~\cite{smith:scattering}, design circuit based quantum computers~\cite{averin:qcomp,babcock:adiabatic,vager:fastadiabatic}
and has been used to create a new paradigm for quantum computing; namely adiabatic quantum computing~\cite{farhi:adiabatic,aharonov:adequiv}.

In many of these applications, bounding the time needed to apply the adiabatic approximation was seen as a secondary issue.  Adiabatic quantum computing proved to be a notable exception.  In adiabatic quantum computing, quantum algorithms are constructed by choosing a time-dependent Hamiltonian whose ground state at $t=0$ is an easily prepared quantum state and at $t=T$, the ground state encodes the solution to a computational problem.  If the adiabatic approximation holds, then the time-evolution transforms the initial ground state into the solution to the problem, thereby enacting the quantum algorithm.  A different notion of algorithmic complexity must be used to assess these algorithms because they do not use a discrete set of gates. The scaling of the energy and the time needed for this evolution with the problem size, is a natural way to assess the complexity of such algorithms~\cite{farhi:adiabatic,roland:adsim}.  Consequently, rigorous estimates for the error in the adiabatic approximation are needed to demonstrate such a scaling~\cite{ruskai:adiabatic}.

 There are many possible ways to adiabatically evolve the initial state into the final state~\cite{roland:localad,farhi:adiabaticpaths,rezakhani:adiabaticexponential}.
We are, at least in principle, given the freedom to choose these Hamiltonians to depend on the properties of the evolution such as the Hilbert space dimension, the error tolerance or even the total evolution time.  This line of inquiry lead to the
Marzlin--Sanders counterexample~\cite{marzlin:counter}, which overturned conventional wisdom about the validity of the approximation.

In 2004 Marzlin and Sanders~\cite{marzlin:counter} put forward a Hamiltonian that
satisfied the traditionally held  ``slowness'' criterion for the validity of the approximation, and
yet it generated an evolution that did not agree with the adiabatic approximation, even in the limit of arbitrarily
slow evolution.  This
created doubt in the scientific community about the conditions that are necessary for the adiabatic approximation and, along with the discovery of adiabatic quantum computing, has led to further interest into the conditions are needed to apply the adiabatic approximation~\cite{wu:adiabatic,mackenzie:adiabaticvalidity,lidar:adiabaticvalidity}.

We achieve three important goals in this paper.  First, we provide a rigorous analysis of adiabatic evolution
that is intuitive, elementary and avoids ad hoc physical assumptions.  Second, we provide
upper bounds and lower bounds for the error in the adiabatic approximation that are often asymptotically tight, in the limit of slow evolution for a fixed Hamiltonian.  Third and last, we show that the adiabatic approximation can be successfully applied
to a broad class of quantum systems that are related to the Marzlin--Sanders counterexample.

\section{Background}
In this section we present the background material that is needed to understand our main results,
which are given in the following section.  In particular we discuss the
adiabatic approximation in greater detail, the Marzlin--Sanders
counterexample~\cite{marzlin:counter} and the Jansen, Ruskai and Seiler (JRS) bound on the error in the adiabatic approximation~\cite{ruskai:adiabatic}.

The time-evolution of closed quantum systems over a time interval $t\in[0,T]$ is described by the time-evolution operator, which is expressed
in terms of the dimensionless time $s=t/T$ as
\begin{equation}
U(s,0)\ket{\psi(0)}:=\ket{\psi(s)},
\end{equation}
where $\ket{\psi(s)}$ is a quantum state evaluated at the dimensionless time $s$.

We write the adiabatic approximation for the action of the time-evolution operator on an instantaneous eigenstate of $H(0)$, $\ket{\gst(0)}$ (typically taken to be the ground state) as
\begin{equation}
U(s,0)\ket{\gst(0)}\approx e^{-i\int_0^sE_{\gst}(\xi)\mathrm{d}\xi T}\ket{\gst(s)},\label{eq:adapprox}
\end{equation}
where $E_{\gst}(s)$ is the eigenvalue of the instantaneous eigenvector $\ket{\gst(s)}$, and the geometric phase~\cite{berry:phase} is absorbed into the definition of the eigenstate.
The error in the adiabatic approximation is the
magnitude of the projection of the evolved state onto the space that is orthogonal to
$\ket{\gst(s)}$.  This error is generally considered to be small if the variation of the Hamiltonian with time
is small compared to the square of the minimum eigenvalue gap between the initial state $\ket{G(s)}$, and all other instantaneous eigenvectors~\cite{kato:adiabatic,sakurai}.  We write this requirement as
\begin{equation}
\frac{\max_s\|\dot H(s)\|}{\min_{s,\nu}[\gap_{\nu,\gst}(s)]^2}\in o(T),\label{eq:adiabaticcond}
\end{equation}
where $\gap_{\nu,\gst}(s)$ is the eigenvalue gap between the instantaneous eigenstates $\ket{\nu(s)}$ and $\ket{\gst(s)}$.  Adiabatic theorems have also been proven for systems that have eigenvalue crossings, but approach those crossings in a smooth fashion~\cite{avron:adiabatic}.   

The Marzlin--Sanders counterexample~\cite{marzlin:counter} provided a Hamiltonian for which~\eqrefb{eq:adiabaticcond}
holds but~\eqrefb{eq:adapprox} does not.  Their
counterexample Hamiltonian is chosen such that $\|\ddot H(s)\|$ diverges in the limit as the quantum system's evolution becomes arbitrarily slow.  The discrepancy with~\eqrefb{eq:adapprox} was attributed to a resonant effect, wherein the normally
negligible contribution of the second derivative of $H$ became dominant due
to increasingly rapid variation as $T$ increases.  Although the fact that the error in the adiabatic approximation depends on the differentiability class of the Hamiltonian was recognized before the Marzlin-Sanders counterexample was discovered~\cite{avron:adiabatictheorem,reichardt:adiabatic,teufel:adiabatic}, the counterexample showed that terms that depend on higher-order derivatives of $H(s)$ can sometimes dominate the expression for the error in the adiabatic approximation.

This problem can be addressed by using upper bounds for the error
in the adiabatic approximation, rather than asymptotic estimates, because
upper bounds do not neglect any contribution to the error.  Examples of such upper bounds can be found in~\cite{ambainis:adiabatic,reichardt:adiabatic,teufel:adiabatic,ruskai:adiabatic}.
The best upper bound presently known for the error in the adiabatic approximation is given by the JRS bound~\cite{ruskai:adiabatic}. The JRS bound states that for an $m$-fold degenerate initial state $\ket{\gst(0)}$ (meaning that there exist $m-1$ instantaneous eigenvectors that share the same eigenvalue as $\ket{\gst(s)}$ for all $s$), the approximation error $\|(\openone-P_\gst(1))U(1,0)\ket{\gst(0)}\|$ is bounded above by
\begin{align}
 \frac{m\|\dot H(0)\|}{\gap(0)^2T}+\frac{m\|\dot H(1)\|}{\gap(1)^2T}+\frac{1}{T}\int_0^1\left(\frac{m\|\ddot H(\xi)\|}{\gap(\xi)^2} +\frac{7m\sqrt{m}\|\dot H(\xi)\|^2}{\gap(\xi)^3}\right)\mathrm{d}\xi.\label{eq:ruskaibd}
\end{align}
The quantity $\gap(\xi)$ in~\eqrefb{eq:ruskaibd} represents the minimum eigenvalue gap between the state $\ket{\gst(\xi)}$
and any other instantaneous eigenstate at time $s=\xi$, $P_\gst(1)$ represents the projector $\ket{\gst(1)}\bra{\gst(1)}$ and $\openone$ is the identity operator.   Given that the Hamiltonian $H(s)$, is not explicitly a function of $T$, this upper bound can be inverted to find a value of $T$ such that the error in the adiabatic approximation is less than any fixed error tolerance.

As~\eqrefb{eq:ruskaibd} depends on $\|\ddot H(s)\|$, it seems to be ideal for addressing the issues caused by
the divergent second derivative of the Marzlin--Sanders counterexample Hamiltonian.  However, JRS noted that the third term in~\eqrefb{eq:ruskaibd} can be replaced by one that is $O(1/T^2)$~\cite{ruskai:adiabatic,reichardt:adiabatic}, implying that their upper bound is not asymptotically tight.  Because their bound is not asymptotically tight, it cannot be used to find a sufficient condition for the validity of~\eqrefb{eq:adiabaticcond}.

We expand on these results by providing upper and lower bounds for the error in the
adiabatic approximation that are often asymptotically tight (meaning that our error bounds converge
to the observed error in the limit of large $T$), given that the leading order term in
our bounds does not vanish for any $T$.  We then use these bounds to address Hamiltonians that are
related to the Marzlin--Sanders counterexample.  We find using these bounds that many Hamiltonians whose
second derivatives diverge as $T$ increases, obey the adiabatic approximation.  These results
are summarized in the following section.


\section{Results}
This section contains our main results, which address the question of when the adiabatic approximation can be used and
provide asymptotic estimates of the error in the approximation.  These results are consequences of our upper bounds
for the error in the adiabatic approximation, which we present in Section~\ref{sec:approximatepath}.


As with most adiabatic theorems, we require that the variation of the Hamiltonian is slow compared to an appropriate timescale.  
Our first theorem states that the adiabatic approximation is valid if the evolution time $T$, is large compared to a timescale $\Delta_0$, which we define below. Our timescale $\Delta_0$ reduces to that in~\eqrefb{eq:adiabaticcond} if $T$ is large compared to a second timescale $\Delta_1$, which depends on the second and third derivatives of the Hamiltonian.  These timescales are given in the following definition.

\begin{definition}\label{def:timescale}
Let $H(s)$ be differentiable three times on the interval $[0,1]$, $\gap_{\nu,\gst}(s)$ be the eigenvalue gap between the instantaneous eigenstates $\ket{G(s)}$ and $\ket{\nu(s)}$ and let $\gap_{\min}$ be the minimum eigenvalue gap
between any two non-degenerate instantaneous eigenvectors of $H(s)$ on this interval. We then define the timescales $\Delta_0$
and $\Delta_1$ to be
\begin{align}\Delta_0:=&\frac{\|\dot H\|}{\min_{\nu,s}[\gap_{\nu,G}(s)]^2}+\frac{\Delta_1}{\gap_{\min}T},\\
\Delta_1:=&\frac{1}{\gap_{\min}}\left(\frac{\|\dddot H\|}{\gap_{\min}}+\frac{\|\ddot H\|^3}{\gap_{\min}^3}+\frac{\|\dot H\|^6}{\gap_{\min}^6} \right),\label{eq:delta0def}
\end{align}
where $\|\dot H\|$, $\|\ddot H\|$ and $\|\dddot H\|$ represent the maximum values of the norms of the first three derivatives of $H(s)$.
\end{definition}

The following theorem uses the criterion that the timescale $\Delta_0$ is asymptotically smaller than $T$ as a sufficient condition for the applicability of the adiabatic approximation, in the limit of large $T$.  We refer to the resulting theorem as a zeroth-order adiabatic theorem because it gives an expression for the approximation error that is correct to zeroth-order in powers of $T^{-1}$.

\begin{theorem}\label{thm:zeroorder}(Zeroth-order adiabatic theorem)
Let $H:[0,1]\mapsto \mathbb{C}^{N\times N}$ be a Hamiltonian that is differentiable three times and has a minimum eigenvalue gap between
any two non-degenerate eigenvectors of $\gap_{\min}>0$.
If $\ket{\gst(s)}$ is a non-degenerate instantaneous eigenstate of $H(s)$ and $\Delta_0\in o(T)$, then the error in the adiabatic approximation obeys
 \begin{equation}
(\openone-P_G(1))U(1,0)\ket{G(0)}= O(\Delta_0/T),
 \end{equation}
 where $\openone$ is the identity operator.
\end{theorem}

The above theorem states that the error in the adiabatic approximation is zero, plus an additional term that is small in the limit of large $T$.  The requirement that $\Delta_0 \in o(T)$, reduces to the standard criterion given in~\eqrefb{eq:adiabaticcond} if the derivatives of $H(s)$ are independent of $T$.  This implies that Theorem~\ref{thm:zeroorder} is as a generalization of the standard criterion for the validity of the approximation.  The following theorem improves upon the result of Theorem~\ref{thm:zeroorder} by providing an expression for the error in the adiabatic approximation that is correct to leading order in powers of $T^{-1}$.  We call this result a first-order adiabatic theorem.

\begin{theorem}\label{thm:jumpbound}(First-order adiabatic theorem)
Let $H:[0,1]\mapsto \mathbb{C}^{N\times N}$ be a Hamiltonian that is differentiable three times and has a minimum eigenvalue gap between
any two non-degenerate eigenvectors of $\gap_{\min}>0$.
If we take $\ket{\gst(s)}$ to be a non-degenerate instantaneous eigenstate of $H(s)$ and $\Delta_1\in o(\gap_{\min}T^2)$, then the error in the adiabatic
 approximation, $(\openone-P_{\gst}(1))U(1,0)\ket{\gst(0)}$, is
\begin{align}
\left.\sum_{\nu\neq \gst}e^{-i\phi_\nu}\ket{\nu(1)}\left(\frac{\braket{\dot\nu(s_1)}{\gst(s_1)}e^{-i\int_0^{s_1} \gap_{\gst,\nu}(\xi)\mathrm{d}\xi T}}{-i\gap_{\gst,\nu}(s_1)T}\right)\right|_{s_1=0}^{1}+O\left(\frac{\Delta_1}{\gap_{\min}T^2}\right)\label{eq:adAssScale},
\end{align}
where $\phi_\nu=\int_0^1E_{\nu}(\xi)\mathrm{d}\xi T$, $\gap_{\gst,\nu}(s):=E_{\gst}(s)-E_{\nu}(s)$,  $P_{\gst}(1)=\ket{\gst(1)}\bra{\gst(1)}$ and the phase of each instantaneous eigenstate is chosen such that $\braket{\dot \nu(s)}{\nu(s)}=0$ for every $s\in[0,1]$ and every $\nu\in \{0,\ldots,N-1\}$.
\end{theorem}
It may seem strange that the leading order expression in~\eqrefb{eq:adAssScale} only depends on the properties of the Hamiltonian at the beginning and end of the adiabatic evolution, in contrast to~\eqrefb{eq:adiabaticcond}.  This dependence occurs because the time-evolution operator is expressed as a sum of rapidly oscillating path-integrals in the adiabatic limit~\cite{mackenzie:adiabatic}.  We expect that the contribution of each complete oscillation will have a negligible contribution to the integral if the oscillation frequency is sufficiently rapid compared to the timescale for the variation of the Hamiltonian.  Therefore we expect that the contribution of the \emph{incomplete} oscillations will dominate the behavior of the integral.  Only the first and last such oscillations are not necessarily complete.  The contribution of the first the last and last oscillation is dominated by the properties of the Hamiltonian at the beginning and the end of the evolution respectively; therefore, we expect that the leading order expression for the error in the adiabatic approximation should be dictated by the properties of the Hamiltonian at those times.  This dependence is also noted by a number of other studies~\cite{mackenzie:adiabatic,rezakhani:adiabaticexponential,ruskai:adiabatic,sun:higherorder_adiabatic}.

These theorems imply that the adiabatic approximation also applies for some Hamiltonians that behave like the Marzlin--Sanders counterexample in that $\|\ddot H\|$ and $\|\dddot H\|$ are increasing functions of $T$.  Our first corollary, which is given below, provides a criterion for the size of the derivatives of $H$ that is sufficient to guarantee that the standard estimate of the error in the adiabatic approximation applies.



\begin{corollary}\label{cor:marzsand2}
If $H:[0,1]\mapsto \mathbb{C}^{N\times N}$ is differentiable three times, $\ket{\gst(s)}$ is a non-degenerate eigenstate of $H(s)$, 
$\|\ddot H\|\in o( \sqrt{T})$, $\|\dddot H\|\in o({T})$
and both $\|\dot H\|$ and $(\gap_{\min})^{-1}$ are bounded above by a constant function of $T$ for all $T\ge 0$ then
\begin{equation}
\|(\openone-P_{\gst}(1))U(1,0)\ket{\gst(0)}\|\in O\left(\max_{s=0,1}\left[\frac{\|\dot H(s)\|}{\min_{\nu}[\gap_{\nu,G}(s)]^2T}\right]\right).
\end{equation}
\end{corollary}

In other cases the adiabatic approximation may hold, but the error in the adiabatic approximation may not be $O(1/T)$.  A sufficient condition for the validity of the adiabatic approximation, for a sufficiently large $T$, is given below.
\begin{corollary}\label{cor:marzsand}
If $H:[0,1]\mapsto \mathbb{C}^{N\times N}$ is twice differentiable, $\ket{\gst(s)}$ is a non-degenerate eigenstate of $H(s)$, 
$\|\ddot H\|\in o(T)$
and both $\|\dot H(s)\|$ and $(\gap_{\min})^{-1}$ are bounded above by a constant function of $T$ for all $T\ge 0$ then

\begin{equation}
\lim_{T\rightarrow \infty}\|(\openone-P_{\gst}(1))U(1,0)\ket{\gst(0)}\|=0.
\end{equation}
\end{corollary}

A principal limitation of these results is that we take $\ket{G(s)}$ to be non-degenerate for all $s$.  We make this assumption to simplify our analysis of the time-evolution operator, and leave the general case for future work.

These results are laid out as follows: in Sec.~\ref{sec:paths} we discuss the concept of paths
in greater detail and their relation to time-evolution.  In Sections~\ref{sec:pathint} and~\ref{sec:recursivepath} we
use the results of Sec.~\ref{sec:paths} to provide a path integral representation of the
time-evolution operator.  In Sec.~\ref{sec:approximatepath} we use integration by parts
to approximate the path integrals and prove Theorems~\ref{thm:zeroorder} and~\ref{thm:jumpbound}.
We then compare our error bounds to the error found for the Search Hamiltonian in Sec.~\ref{sec:searchham}
and discuss the Marzlin--Sanders counterexample~\cite{marzlin:counter} and prove Corollaries~\ref{cor:marzsand} and~\ref{cor:marzsand2} in Sec.~\ref{sec:marzlinsanders}.

\section{Paths}\label{sec:paths}
The concept of paths was first introduced by Wiener in 1924~\cite{wiener:classical}, in which he showed that the probability distribution for
a classical system is found by summing the final states of all paths that could describe the evolution of the system.
The corresponding result for quantum systems was proven by Feynman in 1948~\cite{feynman:pathint}, and a version that uses the Hamiltonian, rather the Lagrangian, is given
 by Farhi and Guttman~\cite{farhi:pathintegral}.  The latter approach is used by
 Mackenzie, Marcotte and Paquette (MMP)~\cite{mackenzie:adiabatic} to find an expression for the time-evolution of adiabatic quantum systems using perturbation theory.  We provide the background material that is needed to understand what we mean by paths in this section, and in the subsequent section, we
use these results to derive the MMP representation of the time-evolution operator, but without
using perturbation theory.

Our usage of path integrals is somewhat different than Feynman's original formulation.  We obtain our path integral representation, similarly to Feynman, by first discretizing time.  By approximating the Hamiltonian with one that is piecewise constant on these discrete time steps, we obtain an approximate evolution operator that is a product of exponentials of the Hamiltonian evaluated at each of the discrete times.  Unlike Feynman, we switch to the instantaneous eigenbasis of the Hamiltonian at the given time-step.  The resulting expression is then simplified by expanding the product of sums into a sum of products of the projectors that are used to represent the basis transformations.  Each product of projectors that occurs in this expansion can then be interpreted as a description of \emph{a potential} evolution of the system, which we call a path.  We then take the limit as the number of time-steps in our discrete time-evolution operator approaches infinity.  This removes the approximation error caused by taking the Hamiltonian to be piecewise constant and turns the sum over all such paths into a sum of path integrals.  These ideas will be elaborated on as we proceed.


We now present a more formal discussion what we mean by a path.
A path is a description of a sequence of states that could describe a sequence of observations
of the evolution of the quantum state.  In particular, we take a path to be an ordered pair $(\{\nu\},\{s\})$, where $\{\nu\}$ is a list of labels that describe the instantaneous eigenstates of $H(s)$ and $\{s\}$ specifies the time that the system transitions to that state.

We require that no two consecutive labels in a path are identical, ie., $\nu_q\ne \nu_{q+1}$ for all $q$.
With this restriction, it is clear that the elements of $\{s\}$ correspond to the times at which the system jumps from one state to another.  We also take the sequence $s$ to be a strictly increasing sequence of dimensionless times taken from the interval $[0,1]$.  We take these times to be strictly increasing because paths that move backwards in time, or are in two different states at the same time, cannot describe a physically realistic sequence of measurements of the state in its instantaneous eigenbasis.    As a final note, we first consider the case where each path is composed of at most $L+1$ distinct states and times, where the times in $\{s\}$ are taken to be integer multiples of $1/L$ to simplify the analysis, similarly to Ambainis and Regev~\cite{ambainis:adiabatic}.  Unlike the formulation of their result, we take the limit as $L$ approaches infinity early in our analysis thereby removing $L$ from our final results.

As an example of a path, consider the evolution for a three-state system with $L=4$ given in Fig.~\ref{fig:jumpdiagram}a: $\ket{0(0)}\rightarrow\ket{1(1/2)}\rightarrow\ket{1(1)}$.  This sequence of states can be recorded as a path that is initially  in the state $\ket{0(0)}$ and jumps to the state $\ket{1(1/2)}$ at time $s=1/2$.
The resulting path $(\{\nu\},\{s\})$ is given by the sequences $\{\nu\}=\{0,1\}$ and $\{s\}=\{0,1/2\}$.
Equivalently, this path can be expressed as a diagram, as seen in Fig.~\ref{fig:jumpdiagram}a.

Diagrams are a convenient way to think about the paths that describe adiabatic evolution.  In our diagrams, an arrow denotes the evolution of the state as a function of time.  If the arrow runs along a line, which denotes a particular energy level, then it implies that the system remains in that instantaneous eigenstate throughout that portion of the evolution. If the arrow connects two energy levels, then it implies that a transition, or jump, between energy levels occurs at that time.  As level crossings are forbidden in our analysis, we do not change the relative spacing between the energy levels as the eigenvalue gap changes with time.  Such diagrams are useful because they convey the intuition of what is being described by our more mathematical notation, which uses a tuple of sequences to describe the path.  A similar diagrammatic notation is also presented by MMP~\cite{mackenzie:adiabatic}.

In this work we classify paths by the number of jumps that they contain. As an example, consider the following evolution: $\ket{0(0)}\rightarrow\ket{1(1/2)}\rightarrow\ket{2(3/4)}\rightarrow\ket{0(1)}$.  The corresponding path is parameterized by $\{\nu\}=\{0,1,2,0\}$ and $\{s\}=\{0,1/2,3/4,1\}$.  This path jumps from one level to another three times and therefore we say that it contains three jumps.  Similarly, the previous example contained one jump.  The number of jumps present in a path is $|\{\nu\}|-1$.  If $(\{\nu\},\{s\})$ is a $q$-jump path then we define $\{\nu\}:=\{\nu_0,\nu_1,\ldots,\nu_q\}$ and similarly $\{s\}:=\{s_0,s_1,\ldots,s_q\}$, where $\nu_0=\gst$ and $s_0=0$.  We choose $\nu_0$ and $s_0$ in this fashion so that the \Th{p} jump occurs at time $s_p$, which makes our results easier to read.  A
diagram of the path in this example is given in Fig.~\ref{fig:jumpdiagram}b.

Additionally, some paths describe adiabatic evolution and others do not.  We call an $n$-jump
path ``adiabatic'' if the final and initial state in the path are the same.  For example, the path in Fig.~\ref{fig:jumpdiagram}b,
is an adiabatic path whereas the path in Fig.~\ref{fig:jumpdiagram}a is not.  We call any path that is not ``adiabatic''
a ``non-adiabatic'' path.  We use the notation $J_q$ to represent the set of all $q$-jump paths
 and $J_q'$ to represent the set of all $q$-jump non-adiabatic paths.
Non-adiabatic paths are crucially important to our work because the
error in the adiabatic approximation is given by the sum over all non-adiabatic paths.  This observation is especially relevant in light of the fact that the convergence problems noted by MMP~\cite{mackenzie:adiabatic} occur because their sums are not restricted to non-adiabatic paths.

\begin{figure}[t!]
\centering
{\includegraphics[width=1\textwidth]{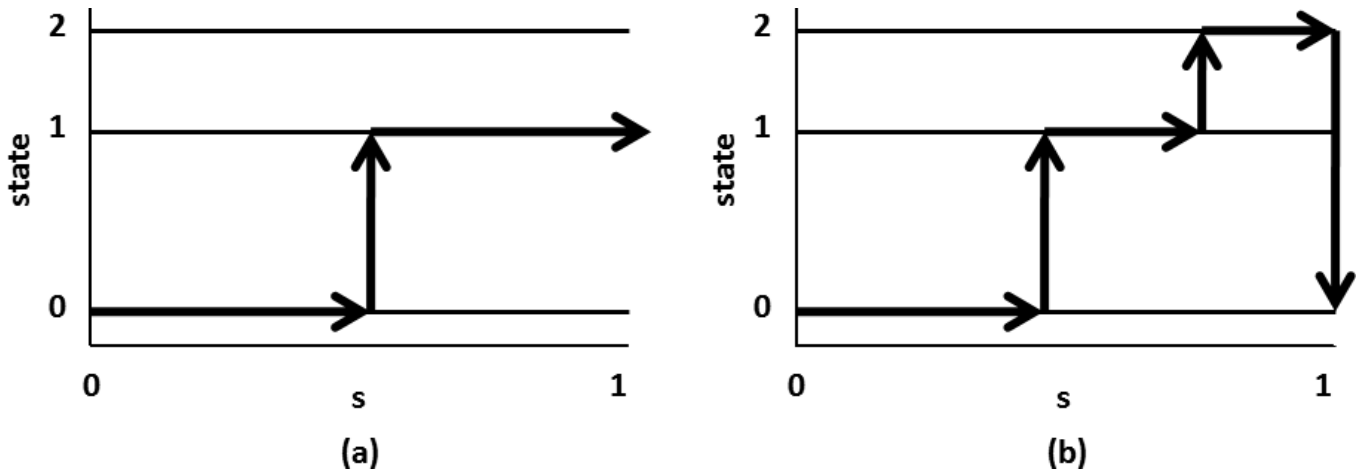}}
\caption[Examples of diagrammatic representations for path integrals.]{These diagrams are a representation of the paths used in the examples in this section.  Here the state labels refer to different values of $\nu$ and do not necessarily refer to energy.\label{fig:jumpdiagram}}
\end{figure}

We describe below, the sum over all such paths in the limit of continuous time.

\begin{definition}\label{def:nonadmeas}
Let $(\{\nu\},\{s\})$ represent an arbitrary path in $J_q$, in the limit as $L\rightarrow \infty$.  We then define the sum over the set all $q$-jump paths, $J_q$ to be
\begin{equation}
\int\mathrm{d}\left(J_q \right):=\sum_{\nu_1\ne \gst}\sum_{\nu_2\ne\nu_1}\cdots\sum_{\nu_{q-1}\ne\nu_q}\int_{0}^{1}\cdots\int_{0}^{s_3}\int_{0}^{s_2}\mathrm{d}s^{q},
\end{equation}
and similarly we define the sum over all $q$-jump non-adiabatic paths, $\Jfunc_q$, to be
\begin{equation}
\int\mathrm{d}\left(\Jfunc_q \right):=\sum_{\nu_1\ne \gst}\sum_{\nu_2\ne\nu_1}\cdots\sum_{{\nu_{q}\not\in \{\nu_{q-1},\gst\}}}\int_{0}^{1}\cdots\int_{0}^{s_3}\int_{0}^{s_2}\mathrm{d}s^{q},
\end{equation}
where $\mathrm{d}s^q=\mathrm{d}s_1\cdots\mathrm{d}s_q$.
\end{definition}

Now that we have introduced our notions of paths and the sum over all paths, we apply them to describe the time-evolution of quantum mechanical systems.

\section{Path Integral Representation of time-evolution}\label{sec:pathint}
In this section we show that by taking the continuous limit of a discrete approach that is similar to that of Ambainis and Regev~\cite{ambainis:adiabatic},
we obtain a path integral representation of the time-evolution operator that is similar to that of Mackenzie, Marcotte and Paquette (MMP)~\cite{mackenzie:adiabatic},
but without using perturbation theory.  The convergence issues noted in~\cite{mackenzie:adiabatic} are ameliorated by noting that not all $n$-jump
paths contribute to the error in the adiabatic approximation; however, the sum of all $n$-jump non-adiabatic paths does.  Our goal in this section
is to find an expression for the sum of all the non-adiabatic paths to the error in the adiabatic approximation.  Upper and lower
bounds for the contribution of these paths are proven in Sec.~\ref{sec:approximatepath}.


To simplify our proof, we choose the eigenstates of the Hamiltonian such that
\begin{equation}
\beta_{n,m}(s):=\braket{\dot n(s)}{m(s)}=\begin{cases} 0, &\text{if } E_n(s)\equiv E_m(s)\\
\frac{\bra{n(s)}\dot H(s) \ket{m(s)}}{E_n(s)-E_m(s)}, &\text{otherwise}\end{cases}\label{eq:betadef}
\end{equation}
similarly to~\cite{ambainis:adiabatic,mackenzie:adiabatic}.  Here $\beta_{n,m}(s)$ can be interpreted as a tunneling amplitude per unit of dimensionless time between the states $\ket{n}$ and $\ket{m}$.  The choice $\beta_{n,n}(s)=0$ corresponds to choosing the phases of the instantaneous eigenstates (which are dynamically irrelevant~\cite{balbook}) such that the geometric phase~\cite{berry:phase} is zero, which simplifies our proof dramatically.  We use this choice of eigenstates to state in the following theorem, which is the main result of this section.

\begin{theorem}\label{thm:pathint}
If $\|H(s)\|$ is bounded and $H(s)$ is piecewise-differentiable on the interval $[0,1]$ then $U(1,0)\ket{\gst(0)}$ is
\begin{align}
U(1,0)&\ket{\gst(0)}=\nn
&e^{-i\int_0^1E_\gst(\xi)\mathrm{d}\xi T}\ket{\gst(1)}+\int \left(e^{-i(\int_{s_1}^1E_{\nu_1}(\xi)\mathrm{d}\xi+\int_{0}^{s_1}E_{\gst}(\xi)\mathrm{d}\xi)T}\beta_{\nu_1,\gst}(s_1)\right)\ket{\nu_1(1)}\mathrm{d}\left(J_1 \right)\nonumber\\
&+\int\left(e^{-i(\int_{s_2}^1E_{\nu_2}(\xi)\mathrm{d}\xi+\int_{s_1}^{s_2}E_{\nu_1}(\xi)\mathrm{d}\xi+\int_{0}^{s_1}E_{\gst}(\xi)\mathrm{d}\xi)T}\beta_{\nu_2,\nu_1}(s_2)\beta_{\nu_1,\gst}(s_1)\right)\ket{\nu_2(1)}\mathrm{d}\left(J_2 \right)+\cdots,\label{eq:transmat8}
\end{align}
where $(\{\nu\},\{s\})$ represents an arbitrary path in $J_q$ and $\nu_i$ and $s_i$ refer to particular element in the sequences, and $\beta$ is defined in~\eqrefb{eq:betadef}.
\end{theorem}
The series in \eqrefb{eq:transmat8} has an intuitive interpretation.  It states that the time-evolution is a weighted sum over all paths that
could describe the time-evolution of the quantum system, wherein the weight of a particular path is given by the products of the
 tunneling amplitudes multiplied by the dynamic phase picked up in that path.  Adiabatic behavior occurs because in the limit of large $T$, the phases that are achieved by each path connecting $\ket{\gst(0)}$ and $\ket{\nu(1)}$ for $\nu\ne \gst$, vary considerably.  Similarly, we expect that the contribution of the sum of all paths that begin in $\ket{\gst(0)}$ and end in $\ket{\gst(1)}$ to be suppressed given that the path contains two or more jumps.
 Hence we expect that interference effects will cause their contribution to become diminished in this limit, leaving only the
 zero jump path which links $\ket{\gst(0)}$ to $\ket{\gst(1)}$.  This phase cancelation effect is demonstrated in Fig.~\ref{fig:combphase}.
 \begin{figure}[t!]
\centering
{\includegraphics[width=1\textwidth]{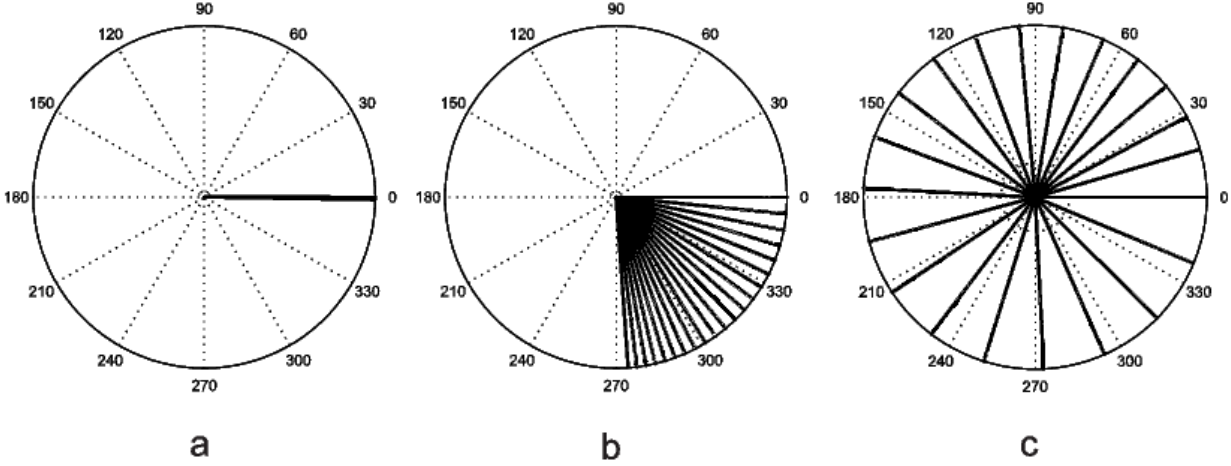}}
\caption[Phase interference effects for one-jump paths.]{This figure is a plot of the phase of states output by $21$ one-jump paths that jump at time $s_1=\{0,1/20,2/20,\ldots,1\}$
from the ground state and the first excited state for a Hamiltonian with $E_\gst(s)=0$ and $E_1(s)=1+s$, for $T=0.01$, $T=1$ and $T=4$ in a, b and c respectively.  Here each solid line is a unit vector in $\mathbb{C}$ that is oriented in the direction of $\exp(-i\int_{s_1}^1 E_1(\xi)-E_\gst(\xi)\mathrm{d}\xi T)$.  This plot shows that as $T$ increases the vectors tend to become uniformly distributed throughout the unit circle, suggesting that the sum over all paths becomes less substantial as $T$ increases.\label{fig:combphase}}
\end{figure}

We use these properties of the eigenstates
to show the following lemma, which is used to evaluate the products of projectors that appear in our
proof of Theorem~\ref{thm:pathint}.
 \begin{lemma}\label{lem:projector}
Let $(\{\nu\},\{s\})$ be a path in $J_q$ with $\{s\}=\{0,j_1/L,\ldots,j_{q}/L\}$ for an integer value of $L$ then,
\begin{align}
&\left(\prod_{j=j_{q-1}+1}^{j_q}P_{\nu_q}\left(\frac jL\right)\right)\cdots\left(\prod_{j=1}^{j_1}P_{\nu_1}
\left(\frac jL\right)\right)=\nn&\qquad\qquad\ket{\nu_q({j_{q}}/L)}\bra{\nu_1(0)}\Biggr(\frac{\prod_{\ell=1}^{q-1}\beta_{\nu_{\ell+1},\nu_{\ell}}
(\frac{j_{\ell}}L)}{L^{q-1}}+O(j_q/L^{q+1})\Biggr)\label{eq:lemmaproduct1},
\end{align}
where each $P_{\nu_i}(j/L)=\ket{\nu_i(j/L)}\bra{\nu_i(j/L)}$ projects onto the ray spanned by the state $\ket{\nu_i(j/L)}$.
\end{lemma}
\begin{proof}
We prove the lemma by performing induction over $q$, which we use to denote the number of jumps in the paths in question.
Our first step is to prove the base case, $q=1$.
Using Taylor's theorem, and the fact that the eigenstates obey $\braket{\dot \nu_\ell}{\nu_\ell}=0$, it follows that for any $\nu_\ell$ and $s$,
\begin{equation}
\braket{\nu_\ell(s+1/L)}{\nu_\ell(s)}=1+O(1/L^2).\label{eq:lem1taylor}
\end{equation}
This implies that for any $j_\ell$
\begin{equation}
\prod_{j=j_{\ell-1}+1}^{j_{\ell}} P_{\nu_\ell}(j/L)=\ket{\nu_\ell(j_{\ell}/L)}\bra{\nu_\ell([j_{\ell-1}+1]/L)}\left(1+O\left([j_{\ell}-j_{\ell-1}-1]/L^2 \right)\right).\label{eq:lem1base2}
\end{equation}
We demonstrate the base case, $q=1$, of~\eqrefb{eq:lemmaproduct1} by substituting $\ell=1$ and $j_0=0$ into~\eqrefb{eq:lem1base2}.

Now let us assume that the result of Lemma~\ref{lem:projector} applies for the case
\begin{align}
&\left(\prod_{j=j_{q-2}+1}^{j_{q-1}}P_{\nu_{q-1}}\left(\frac jL\right)\right)\cdots\left(\prod_{j=1}^{j_1}P_{\nu_1}
\left(\frac jL\right)\right)\nn&\qquad\qquad=\ket{\nu_{q-1}({j_{q-1}}/L)}\bra{\nu_1(0)}\left(\frac{\prod_{\ell=1}^{q-2}\beta_{\nu_{\ell+1},\nu_{\ell}}
(\frac{j_{\ell}}L)}{L^{q-1}}+O(j_{q-1}/L^{q})\right).\label{eq:lemma1assumption}
\end{align}
By using~\eqrefb{eq:lemma1assumption} and~\eqrefb{eq:lem1base2} we find that
the left hand side of~\eqrefb{eq:lemmaproduct1} becomes
\begin{align}
\ket{\nu_{q}(j_{q}/L)}\braket{\nu_q([j_{q-1}+1]/L)}{\nu_{q-1}({j_{q-1}}/L)}\bra{\nu_1(0)}\left(\frac{\prod_{\ell=1}^{q-2}\beta_{\nu_{\ell+1},\nu_{\ell}}
(\frac{j_{\ell}}L)}{L^{q-1}}+O(j_{q}/L^{q})\right).\label{eq:lemma12ndlast}
\end{align}
Taylor's Theorem and the fact that, from the definition of a path, $\nu_q\ne\nu_{q-1}$ imply,
\begin{align}
\braket{\nu_q([j_{q-1}+1]/L)}{\nu_{q-1}({j_{q-1}}/L)}&=\frac{1}{L}\braket{\dot\nu_q(j_{q-1}/L)}{\nu_{q-1}({j_{q-1}}/L)}+O(1/L^2).\\
&=\frac{\beta_{\nu_q,\nu_{q-1}}(j_{q-1}/L)}{L}+O(1/L^2).\label{eq:lemma1last}
\end{align}
We substitute~\eqrefb{eq:lemma1last} into~\eqrefb{eq:lemma12ndlast}.  The induction step on $q$ follows, completing our proof of the lemma.
\end{proof}

Lemma~\ref{lem:projector} provides an expression that can be used to evaluate products of projection operators.
These results allow us to prove Theorem~\ref{thm:pathint} by discretizing the time-evolution and
examining the evolution in each time step in the basis formed by the instantaneous eigenstates of $H(s)$.  We use projectors to
transform the basis from the previous time step to the basis at the current time step.  The resulting expression for the time-evolution is given as a weighted sum of products of projectors, which we then simplify
by using Lemma~\ref{lem:projector}.  Our proof of Theorem~\ref{thm:pathint} is given below.  

\begin{proofof}{Theorem~\ref{thm:pathint}}
It follows from using the Trotter formula~\cite{trotter} for a fixed value of $L$ that if
$H(s)$ is piecewise differentiable then we can approximate $H(s)$ by a piecewise constant Hamiltonian to find
\begin{align}
U(1,0)&=\prod_{j=0}^L\left(e^{-iH(j/L)T/L} \right)+O(1/L),\nonumber\\
&=\prod_{j=0}^L\left(\sum_{\nu=0}^{N-1}P_\nu(j/L)e^{-iE_\nu(j/L)T/L} \right)+O(1/L),\label{eq:productsum}
\end{align}
where for any $s\in[0,1]$ and $\nu\in \{0,1,\ldots,N-1\}$, $\ket{\nu(s)}$ is an eigenstate of $H(s)$ and $P_\nu(s):=\ket{\nu(s)}\bra{\nu(s)}$.
We then expand~\eqrefb{eq:productsum} to find

\begin{align}\label{eq:transmat3}
U(1,0)\ket{\gst(0)}&=\sum_{\nu_1=0}^{N-1}\cdots \sum_{\nu_L=0}^{N-1}\prod_{j=1}^Le^{-iE_{\nu_j}(j/L)T/L} P_{\nu_j}(j/L)\ket{\gst(0)} +O(1/L).
\end{align}
Equation~\ref{eq:transmat3} is a weighted sum over every possible
path that can describe the evolution of the quantum system over the time interval.  We group the paths in the sum by the number of jumps
that comprise them and apply the result of Lemma~\ref{lem:projector}
to find,

\begin{align}
U(1,0)\ket{\gst(0)}=&e^{-i(\sum_{j=1}^LE_{0}(j/L)T/L}\ket{\gst(1)}\nn
&+\sum_{\nu_1\neq \gst}\sum_{j_1=1}^L\frac{e^{-i(\sum_{j=j_1}^LE_{\nu_1}(j/L)+\sum_{j=0}^{j_1-1}E_{\gst}(j/L))T/L}\beta_{\nu_1,\gst}(j_1/L)}{L}\ket{\nu_1(1)}\nonumber\\
&+\sum_{\nu_2\neq \nu_1}\sum_{\nu_1\neq\gst}\sum_{j_2=2}^{L}\sum_{j_1=1}^{j_2-1}{e^{-i(\sum_{j=j_2}^LE_{\nu_1}(j/L)+\sum_{j=j_1}^{j_2-1}E_{\nu_1}(j/L)+\sum_{j=0}^{j_1-1}E_{\gst}(j/L))T/L}}\nn &\qquad\times\frac{\beta_{\nu_2,\nu_1}(j_2/L)\beta_{\nu_1,\gst}(j_1/L)}{L^2}\ket{\nu_2(1)}\nonumber\\
&+\cdots+O(1/L).\label{eq:transmat4}
\end{align}
Equation~\ref{eq:transmat4} is similar to that of Theorem~\ref{thm:pathint}, except instead of time-integrals we have discrete sums over time and we also have an $O(1/L)$ term that represents the error in simulating the adiabatic evolution using a piecewise constant Hamiltonian.  We replace many of the sums in~\eqrefb{eq:transmat4} by integrals, using the fact that \begin{equation}
\sum_{j=s_1L}^{s_2L} f(j/L)/L=\int_{s_1}^{s_2}f(s)\mathrm{d}s+O(1/L),\end{equation}
and eliminate the simulation error, by examining~\eqrefb{eq:transmat4} in the limit as $L\rightarrow\infty$.  This demonstrates that the sum of all one- and two-jump paths approach the path integrals in Theorem~\ref{thm:pathint}.  The analogous results for paths that contain three or more jumps follow from the exact same reasoning, completing our proof of the theorem.
\end{proofof}

We have introduced our concept of path-integrals in this section, and have found an expression for the time-evolution operator in terms of these path integrals.
Our next step is to approximate the contribution of the non-adiabatic paths in~\eqrefb{eq:transmat8}. This gives us the error in the adiabatic approximation.
We use the results of the following section to find such a bound, which we give in Sec.~\ref{sec:approximatepath}.


\section{A Recursive Expression for Path Integrals}\label{sec:recursivepath}

Although an upper bound on the contribution one-jump non-adiabatic paths to the error in the adiabatic approximation can be directly calculated from Theorem~\ref{thm:pathint},
it is difficult to upper bound the contribution of paths that contain more than one jump using that strategy.  To do so, we
use the fact that every $m$-jump path is a concatenation of $m$ one-jump paths.  We use this property to find a recursion relation
that allows us to express the contribution of an $m$-jump non-adiabatic path in terms of the contribution of an $(m-1)$-jump non-adiabatic path (this fact follows from our definition of paths that we introduce in Sec.~\ref{sec:paths}).
Our error bounds then follow from summing a geometric series that results from unfolding these recursion relations.

Using the series in~\eqrefb{eq:transmat8} and neglecting all paths that end in $\gst$ (because such paths do not contribute to the approximation error), the error in the adiabatic approximation can be expressed as
\begin{equation}
\big(\left(\openone-P_{\gst}(1) \right)U(1,0)\big)\ket{\gst(0)}=\sum_{q=1}^\infty \left(\sum_{\nu_1\ne \gst}\sum_{\nu_2\ne\nu_1} \cdots\sum_{\nu_{q}\not\in\{\nu_{q-1},\gst\}} e^{-i\int_0^1E_{\nu_q}(s)\mathrm{d}sT}\Ffunc_{\nu_q}(1)\ket{\nu_q(1)}\right),
\end{equation}
where $\gap_{\mu,\nu}(s)$ represents the eigenvalue gap between the instantaneous eigenstates $\ket{\mu(s)}$ and $\ket{\nu(s)}$ and $\Ffunc_{\nu_q}$ represents the sum over all times as given by
\begin{align}
\Ffunc_{\nu_q}(1)&=\int_{0}^{1}e^{-i\int_{0}^{s_q}\gap_{\nu_{q-1},\nu_q}(\xi)\mathrm{d}\xi T}\beta_{\nu_q,\nu_{q-1}}(s_q)\nonumber\\
&\times\int_{0}^{s_{q}}e^{-i\int_{0}^{s_{q-1}}\gap_{\nu_{q-2},\nu_{q-1}}(\xi)\mathrm{d}\xi T}\beta_{\nu_{q-1},\nu_{q-2}}(s_{q-1})
\times\cdots\nonumber\\
&\times\int_{0}^{s_{2}}e^{-i\int_{0}^{s_1}\gap_{\gst,\nu_1}(\xi)\mathrm{d}\xi T}\beta_{\nu_1,\gst}(s_1)\mathrm{d}s_1\cdots\mathrm{d}s_{q}.\label{eq:Ffuncdef}
\end{align}
The integrals in~\eqrefb{eq:Ffuncdef} are expressed recursively as
\begin{align}\label{eq:jumpsec2}
\Ffunc_{\nu_q}(x)&=\int_0^x e^{-i\int_{0}^{s}\gap_{\nu_{q-1},\nu_q}(\xi)\mathrm{d}\xi T}\beta_{\nu_q,\nu_{q-1}}(s)\Ffunc_{\nu_{q-1}}(s)\mathrm{d}s,\nonumber\\
\Ffunc_{\nu_1}(x)&=\int_0^x e^{-i\int_{0}^{s}\gap_{\gst,\nu_1}(\xi)\mathrm{d}\xi T}\beta_{\nu_1,\gst}(s)\mathrm{d}s.
\end{align}
A similar recursive expansion to~\eqrefb{eq:jumpsec2} is found by Sun in~\cite{sun:higherorder_adiabatic}.  In the following section we find higher-order approximations to the error in the adiabatic approximation by using integration by parts to approximate the recursion relations in~\eqrefb{eq:jumpsec2}.

\section{Approximating Path Integrals}\label{sec:approximatepath}
We have argued qualitatively that quantum interference
effects are responsible for the emergence of adiabatic behavior in
quantum systems.  In this section we provide quantitative bounds
for this phase cancelation effect.  We first provide an example of how integration by parts can be used to approximate the phase cancelation effect for one-jump paths, and then use the same strategy to build
an asymptotic expansion for the sum over all $q$-jump non-adiabatic paths in powers of $T^{-1}$, for every $q>0$.
We then sum these expressions over all values of $q$ and obtain bounds for the approximation
error.  We then prove Theorem~\ref{thm:jumpbound} using these bounds for the error.


\subsection{Approximating Integrals Using Integration By Parts}\label{subsec:one-jumpexamp}
In this subsection we provide an example that shows how integration by parts can be used to estimate the contribution of all one-jump non-adiabatic paths to the error in the adiabatic approximation.  The purpose of this introduction is to provide a simple example of approximating path integrals that will help explain our approach for bounding the contribution of $q$-jump paths to the approximation error, which we address in the following subsection.

Integration by parts is a technique that is commonly used to approximate the integrals that appear in proofs of the adiabatic approximation~\cite{mackenzie:adiabatic,ruskai:adiabatic}.  Unlike other methods,
such as the method of steepest descent, integration by parts gives an exact expression for the integral.  Specifically,
if $f:\mathbb{R}\mapsto \mathbb{C}$ and $g:\mathbb{R}\mapsto\mathbb{C}$ are differentiable functions and if $\partial_s{g(s)}$ is non-zero in the
region of integration then~\cite{bender:advancedmath}
\begin{equation}
\int_0^x f(s)e^{\lambda g(s)}\mathrm{d}s=\left.f(s)\left(\lambda \diff{g(s)}{s}\right)^{-1}e^{\lambda g(s)}\right|_0^x-\int_0^x\left(\diff{}{s}\left[ f(s)\left(\lambda\diff{ g(s)}{s}\right)^{-1}\right]\right)e^{\lambda g(s)}\mathrm{d}s.\label{eq:intpartsgen}
\end{equation}
Equation~\ref{eq:intpartsgen} gives the leading order behavior of the integral in the limit of large $\lambda$ as well as an exact expression for the remainder term (which can also be approximated using integration by parts) given that $\partial_s g(s)\ne 0$ for all $s\in[0,x]$.  For this reason, we assume that the minimum eigenvalue gap between any two non-degenerate eigenvectors is non-zero because, if it is ever zero, then
integration by parts will fail to approximate~\eqrefb{eq:transmat8}.  In such circumstances, an alternative method would have to be used to analyze the integral.

We approximate the contribution of all one-jump paths to the error in the adiabatic approximation as an example using integration by parts to bound a path integral.
The contribution of all one-jump paths to the approximation error is
\begin{align}
\int \ket{\nu_1(1)}\beta_{\nu_1,\gst}(s)e^{-i\int_0^s \gap_{\gst,\nu_1}(\xi)\mathrm{d}\xi T}\mathrm{d}(J_1')&= \sum_{\nu_1\ne\gst}\ket{\nu_1(1)}\int_0^1 \beta_{\nu_1,\gst}(s)e^{-i\int_0^s \gap_{\gst,\nu_1}(\xi)\mathrm{d}\xi T}\mathrm{d}s.
\end{align}
We apply integration by parts to obtain,
\begin{align}
\sum_{\nu_1\ne \gst} \ket{\nu_1(1)}\left(\frac{\beta_{\nu_1,\gst}(s)e^{-i\int_0^s\gap_{\gst,\nu_1}(\xi)\mathrm{d}\xi T}}{-i\gap_{\gst,\nu_1}(s)T}\Biggr|_{s=0}^1-\int_0^1 \left(\diff{}{s}\left[\frac{\beta_{\nu_1,\gst}(s)}{-i\gap_{\gst,\nu_1}(s) T} \right]\right)e^{-i\int_0^s\gap_{\gst,\nu_1}(\xi)\mathrm{d}\xi T}\mathrm{d}s\right).\label{eq:intparts1j}
\end{align}
Using $\|\dot H\|:=\max_s\|\dot H(s)\|$ and $\max_{s}\|\ddot H(s)\|\in O(1)$ we find that~\eqrefb{eq:intparts1j} has the upper bound
\begin{equation}
\left\|\int \ket{\nu_1(1)}\beta_{\nu_1,\gst}(s)e^{-i\int_0^s \gap_{\gst,\nu_1}(\xi)\mathrm{d}\xi T}\mathrm{d}(J_1') \right\|\le 2\frac{\|\dot H\|}{\min_{\nu_1,s}[\gap_{\nu_1,\gst}(s)]^2 T}+ O(1/T^2),\label{eq:loworderapprox}
\end{equation}
which is a commonly used estimate of the error in the adiabatic approximation.
If the contribution of all non-adiabatic paths with two or more jumps is negligible,
then~\eqrefb{eq:loworderapprox} implies that the standard estimate for the approximation error holds~\cite{mackenzie:adiabatic}.

Higher-order estimates for the contribution of the one-jump terms can be found by applying integration by parts to approximate the remainder term.
The success of this strategy depends on the radius of convergence of the perturbation expansion.  Unfortunately, the series has a zero-radius of convergence.  This implies that arbitrarily high-order adiabatic theorems generated using our technique will only be useful in the limit as $T\rightarrow\infty$.  This is because the repeated derivatives that occur from using integration by parts to obtain an $O(1/T^k)$ approximation to $C_1$ are $O(k!/T^k)$.  The perturbation series for the one-jump path integrals will generally have a radius of convergence of zero as $k\rightarrow\infty$; suggesting that higher-order adiabatic approximations are often only useful in the limit of large $T$, as also noted in~\cite{rezakhani:adiabaticexponential}.

In the subsequent subsection we generalize the ideas presented in this subsection to find bounds for the approximation error that emerges due to \emph{every} non-adiabatic path, as opposed to the results in this subsection, which considered only paths with one jump.

\subsection{Error Bounds For The Adiabatic Approximation}\label{subsec:generaljump}

Our strategy for bounding the error in the adiabatic approximation is to use integration by parts to find upper bounds for the oscillatory integrals present in Theorem~\ref{thm:pathint}.  Our zeroth-order and first-order adiabatic approximations in Theorems~\ref{thm:zeroorder} and~\ref{thm:jumpbound}, directly follow from this upper bound, which is given in the following Theorem.


\begin{theorem}\label{thm:errorbounds}
Let $H:[0,1]\mapsto \mathbb{C}^{N\times N}$ be a Hamiltonian that is differentiable three times and has a minimum eigenvalue gap between
any two non-degenerate eigenvectors of $\gap_{\min}>0$.  Then, for any $T \ge 0$, the error in the adiabatic approximation obeys
\begin{align}
\left\|(\openone-P_{\gst}(1))U(1,0)\ket{\gst(0)}-\left.\sum_{\nu\ne\gst}\frac{e^{-i\phi_\nu}\ket{\nu(1)}\braket{\dot\nu(s_1)}{\gst(s_1)}e^{-i \int_0^{s_1}\gap_{\gst,\nu}(\xi)\mathrm{d}\xi T}}{-i\gap_{\gst,\nu}(s_1)T}\right|_{s_1=0}^{1}\right\|\leq&\frac{\mathbf{R}}{T^2},\label{eq:thmupperbound}
\end{align}
where $\phi_\nu=\int_0^1 E_\nu(\xi)\mathrm{d}\xi T$ and $\mathbf{R}$ is bounded above by
\begin{align}
\mathbf{R}\le&\frac{2\|\ddot H\|+\|\dddot H\|}{\gap_{\min}^3}+\frac{25\|\dot H\|\|\ddot H\|+16\|\dot H\|^2+\|\ddot H\|^2}{\gap_{\min}^4}+\frac{12\|\ddot H\|\|\dot H\|^2+118\|\dot H\|^3}{\gap_{\min}^5}+\frac{36\|\dot H\|^4}{\gap_{\min}^6}\nonumber\\
&+T^2\left[\left(1+\frac{\Gamma}{\|\dot H\|T}\right)(e^{\Gamma/(\gap_{\min}T)}-1)-\frac{\Gamma}{\gap_{\min}T}\right],
\end{align}
and $\Gamma=\left(\frac{6\|\dot H\|^{3}}{\gap_{\min}^{3}}+\frac{\|\dot H\|\|\ddot H\|}{\gap_{\min}^{2}}\right)+\frac{2\|\dot H\|^{2}}{\gap_{\min}^{2}}$, $\|\dot H\|$, $\|\ddot H\|$ and $\|\dddot H\|$ refer to the maximum values of the norms of the
first three derivatives of $H$ and $\gap_{\gst,\nu}(s)=E_{\gst}(s)-E_{\nu}(s)$.

\end{theorem}

This theorem is important because it provides upper and lower bounds that, except in pathological cases, converge to the observed approximation error in the limit of large $T$.  As a result, these bounds are superior to the JRS bound~\cite{ruskai:adiabatic} in the limit of large $T$.  We compare these two bounds in greater detail in Sec.~\ref{sec:searchham}.

The proof of Theorem~\ref{thm:errorbounds} requires several steps, which we express in three lemmas.  
In order to abbreviate the statements of our Lemmas, we define the contribution of all $j$-jump non-adiabatic paths to
the error in the adiabatic approximation to be $C_j$.  Hence the
error is
\begin{equation}
(I-P_{\gst}(1))U(1,0)\ket{\gst(0)}=C_1+C_2+C_3+\cdots,
\end{equation}
where for example,
\begin{equation}
\sumOneJumpxpr:=\sumOneJump,
\end{equation}
and
\begin{equation}
\sumTwoJumpxpr:=\sumTwoJump.
\end{equation}
We have that $\nu_2\ne \gst$ in our expression for $C_2$ because any such path is an adiabatic path and thus does not contribute to the error.
 The following lemma contains our upper bounds for the contribution of $C_j$ for $j>2$, which is an important step towards proving Theorem~\ref{thm:errorbounds}.
\begin{lemma}\label{lem:jumpbound}
Let $H:[0,1]\mapsto \mathbb{C}^{N\times N}$ be a Hamiltonian that is twice differentiable, then the following bound holds
on the error in the adiabatic approximation:
\begin{align}
\|(\openone-P_{\gst}(1))U(1,0)\ket{\gst(0)}-\sumOneJumpxpr\|\leq&\|\sumTwoJumpxpr\|+\left(1+\frac{\Gamma}{\|\dot H\|T}\right)(e^{\Gamma/(\gap_{\min}T)}-1)-\frac{\Gamma}{\gap_{\min}T}.\label{eq:lemupperbound}
\end{align}
Here $\Gamma=\left(\frac{6\|\dot H\|^{3}}{\gap_{\min}^{3}}+\frac{\|\dot H\|\|\ddot H\|}{\gap_{\min}^{2}}\right)+\frac{2\|\dot H\|^{2}}{\gap_{\min}^{2}}$ and $\|\dot H\|$ and $\|\ddot H\|$ refer to the maximum values of the norms of the
first two derivatives of $H$ and $\gamma_{\min}$ is the minimum eigenvalue gap between any two non-degenerate eigenvalues.
\end{lemma}
Our proof of Lemma~\ref{lem:jumpbound} uses integration by parts to approximate $\|C_j\|$ for $j>2$.  To do so we must provide upper bounds
for the derivatives of $\beta$ that appear due to using integration by parts to evaluate the integrals.  The most straightforward way to approximate these derivatives is to differentiate~\eqrefb{eq:betadef} and use the triangle inequality on the result.  However, this simple approach leads to extraneous factors of $N$, which appear because of the sums over $\nu_q$.  We use a trick involving the resolution of unity to provide upper bounds for the
first two derivatives of $\beta$ that do not depend on $N$.  We state this result in the subsequent lemma, and will use it to prove Lemma~\ref{lem:jumpbound}.
\begin{lemma}\label{lem:betaderiv}
Let $H(s)$ be twice differentiable on $s\in I=[0,1]$ and let $s'\in I$, $\chi:J_{q-1}\mapsto \mathbb{C}$ be a bounded function, and for simplicity we define the argument of $\chi$ to be the path $\vec{\nu}:=(\{\nu_j:j=0,\ldots,q-1\},\{s_j:j=0,\ldots,q-1\})$ then
\begin{equation}
\left\|\sum_{\nu_q,\nu_{q-1},\ldots,\nu_1}\dot\beta_{\nu_{q},\nu_{q-1}}(s)\chi(\vec{\nu})\ket{\nu_{q}(s')}\right\|\le \left(\frac{4\|\dot H\|^2}{\gap_{\min}^2}+\frac{\|\ddot H\|}{\gap_{\min}}\right)\left\|\sum_{\nu_{q-1},\nu_{q-2},\ldots,\nu_1} \chi(\vec{\nu})\ket{\nu_{q-1}(s)}\right\|\label{eq:betaderiv}
\end{equation}
and
\begin{align}
&\Biggr\|\sum_{\nu_q,\nu_{q-1},\ldots,\nu_1}\ddot\beta_{\nu_{q},\nu_{q-1}}(s)\chi(\vec{\nu})\ket{\nu_{q}(s')}\Biggr\|\le \left(\frac{44\|\dot H\|^3}{\gap_{\min}^3}+\frac{12\|\dot H\|\|\ddot H\|}{\gap_{\min}^2}+\frac{\|\dddot H\|}{\gap_{\min}}\right)\left\|\sum_{\nu_{q-1},\nu_{q-2},\ldots,\nu_1} \chi(\vec{\nu})\ket{\nu_{q-1}(s)}\right\|.\label{eq:betaderiv2}
\end{align}
\end{lemma}
The function $\chi$ in these bounds is an arbitrary function, for example if we wish to find an upper bound for
$\left\|\sum_{\nu_2,\nu_1}\dot\beta_{\nu_{2},\nu_{1}}(s)\beta_{\nu_{1},\gst}(s)\ket{\nu_{2}(s')}\right\|$,
then we can use~\eqrefb{eq:betaderiv2} to find the upper bound by choosing $\chi(\nu_1,s)=\beta_{\nu_1,\gst}(s)$.
The proof of Lemma~\ref{lem:betaderiv} is given in Appendix~\ref{appendix:proofofbetaderiv}.

We will now use the first result in Lemma~\ref{lem:betaderiv} to prove Lemma~\ref{lem:jumpbound}.
\begin{proofof}{Lemma \ref{lem:jumpbound}}
Using integration by parts on \eqrefb{eq:jumpsec2} yields

\begin{equation}\label{eq:jumpsec6}
\Ffunc_{\nu_q}(x)=\left.\frac{\Ffunc_{\nu_{q-1}}(s)e^{-i\gfunc{s}{\nu_{q-1},\nu_{q}}T}\beta_{\nu_q,\nu_{q-1}}(s)}{-i\gfuncp{s}{\nu_{q-1},\nu_{q}}T}\right|_0^x- \int_0^x \frac{\partial}{\partial_s}\left(\frac{\beta_{\nu_q,\nu_{q-1}}(s)\Ffunc_{\nu_{q-1}}(s)}{-i\gfuncp{s}{\nu_{q-1},\nu_{q}}T}\right)e^{-i\gfunc{s}{\nu_{q-1},\nu_{q}}T} \mathrm{d}s,
\end{equation}
where $\gfunc{s}{\nu_{q-1},\nu_{q}}=\int_0^s\gap_{\nu_{q-1},\nu_{q}}(\xi)\mathrm{d}\xi$ and $\Ffunc_{\nu_q}(x)$ represents the time-integrals in the sum over all $q$-jump non-adiabatic paths.
Using~\eqrefb{eq:jumpsec2} and the Fundamental Theorem of Calculus, we can simplify~\eqrefb{eq:jumpsec6} by substituting
\begin{equation}
\partial_s \Ffunc_{\nu_{q-1}}(s)=e^{-i\gfunc{s}{\nu_{q-2},\nu_{q-1}}T}\beta_{\nu_{q-1},\nu_{q-2}}(s)\Ffunc_{\nu_{q-2}}(s).
\end{equation}
We also have that for any $p>0$, $\Ffunc_{p}(0)=0$.
We then simplify~\eqrefb{eq:jumpsec6} by using the recursion relations
to express every instance of $\Ffunc_{\nu_{q-1}}$ as an integral of $\Ffunc_{\nu_{q-2}}$, to find
\begin{align}
&\Ffunc_{\nu_q}(x)= \Biggr( \frac{e^{-i\gfunc{x}{\nu_{q-1},\nu_{q}}T}\beta_{\nu_q,\nu_{q-1}}(x)}{-i\gfuncp{x}{\nu_{q-1},\nu_{q}}T}\int_0^xe^{-i\gfunc{s_1}{\nu_{q-2},\nu_{q-1}}T}\beta_{\nu_{q-1},\nu_{q-2}}(s_1)\Ffunc_{\nu_{q-2}}(s_1)\mathrm{d}s_1
\nonumber\\
&-\int_0^x \frac{\partial}{\partial_{s_1}}\left(\frac{\beta_{\nu_{q},\nu_{q-1}}(s_1)}{-i\gfuncp{s_1}{\nu_{q-1},\nu_{q}}T}\right)e^{-i\gfunc{s_1}{\nu_{q-1},\nu_{q}}T} \int_0^{s_1}e^{-i\gfunc{s_2}{\nu_{q-2},\nu_{q-1}}T}\beta_{\nu_{q-1},\nu_{q-2}}(s_2)\Ffunc_{\nu_{q-2}}(s_2)\mathrm{d}^2s\nonumber\\
&-\int_0^x \left(\frac{\beta_{\nu_q,\nu_{q-1}}(s)\beta_{\nu_{q-1},\nu_{q-2}}(s)}{-i\gfuncp{s}{\nu_{q-1},\nu_{q}}T}\right)e^{-i(\gfunc{s}{\nu_{q-1},\nu_{q}}+\gfunc{s}{\nu_{q-2},\nu_{q-1}})T} \Ffunc_{\nu_{q-2}}(s)\mathrm{d}s\Biggr).\label{eq:intpartrecur}
\end{align}

We then express the sum over all values of $\nu_1,\ldots,\nu_q$ in~\eqrefb{eq:intpartrecur} as
\begin{equation}
\sum_{\nu_q,\nu_{q-1},\ldots}\Ffunc_{\nu_q}(x)\ket{\nu_q(x)}=g_q(x)-h_q(x),\label{eq:intpartrecur2}
\end{equation}
where for any $q>0$,
\begin{align}
h_q(x):=\sum_{\nu_q,\nu_{q-1},\ldots}\ket{\nu_q(x)}\int_0^x\int_0^{s_1}&\left(\diff{}{s_1}\frac{\beta_{\nu_q,\nu_{q-1}}(s_1)}{-i\gap_{\nu_{q-1},\nu_q}(s_1)T}\right)e^{-i(\kappa_{\nu_{q-1},\nu_q}(s_1)+\kappa_{\nu_{q-2},\nu_{q-1}}(s_2))T}\nn
&\qquad\times\beta_{\nu_{q-1},\nu_{q-2}}(s_2)\Ffunc_{\nu_{q-2}}(s_2)\mathrm{d}^2s,\label{eq:hfuncdef}
\end{align}
where $\mathrm{d}^2s=\mathrm{d}s_2\mathrm{d}s_1$ and
\begin{align}
g_q(x):=&\sum_{\nu_q,\nu_{q-1},\ldots}\ket{\nu_q(x)}\Biggr(\int_0^x\frac{e^{-i(\kappa_{\nu_{q-1},\nu_q}(x)+\kappa_{\nu_{q-2},\nu_{q-1}}(s_1))T}\beta_{\nu_{q},\nu_{q-1}}(x)\beta_{\nu_{q-1},\nu_{q-2}}(s_1)f_{\nu_{q-2}}(s_1)}{-i\gap_{\nu_{q-1},\nu_q}(x)T}\nonumber\\
&-\frac{e^{-i(\kappa_{\nu_{q-1},\nu_q}(s_1)+\kappa_{\nu_{q-2},\nu_{q-1}}(s_1))T}\beta_{\nu_{q},\nu_{q-1}}(s_1)\beta_{\nu_{q-1},\nu_{q-2}}(s_1)f_{\nu_{q-2}}(s_1)}{-i\gap_{\nu_{q-1},\nu_q}(x)T}~\mathrm{d}s_1\Biggr).
\end{align}

We now will proceed to bound $\|g_q(x)\|$.  To do so, we use the definition of $\beta$ in~\eqrefb{eq:betadef} to find following bound, which is valid for any $s',s'',s'''$ and $s''''$.
\begin{align}
&\Biggr\|\sum_{\nu_q,\nu_{q-1},\ldots}\ket{\nu_q(x)}\beta_{\nu_q,\nu_{q-1}}(s')\beta_{\nu_{q-1},\nu_{q-2}}(s'')f_{\nu_{q-2}}(s''')/\gap_{\nu_{q-1},\nu_{q}}(s'''')\Biggr\|\nonumber\\
&\qquad=\left\|\frac{\sum_{\nu_q,\nu_{q-1},\ldots}\ket{\nu_q(x)}\bra{\nu_q(s')}\dot H(s')\ket{\nu_{q-1}(s')}\bra{\nu_{q-1}(s'')}\dot{H}(s'')\ket{\nu_{q-2}(s'')}f_{\nu_{q-2}}(s''')}{\gap_{\nu_{q-1},\nu_q}(s')\gap_{\nu_{q-2},\nu_{q-1}}(s'')\gap_{\nu_{q-1},\nu_{q}}(s'''')}\right\|\nonumber\\
&\qquad\le\frac{\|\dot H\|^2}{\gap_{\min}^3}\left\|\sum_{\nu_q}\ket{\nu_q(x)}\bra{\nu_q(s')}\right\|\left\|\sum_{\nu_{q-1}}\ket{\nu_{q-1}(s')}\bra{\nu_{q-1}(s'')}\right\|\left\|\sum_{\nu_{q-2},\ldots}f_{\nu_{q-2}}(s''')\ket{\nu_{q-2}(s'')} \right\| \nonumber\\
&\qquad\le\frac{\|\dot H\|^2}{\gap_{\min}^3}\left\|\sum_{\nu_{q-2},\ldots}f_{\nu_{q-2}}(s''')\ket{\nu_{q-2}(s'')} \right\|,
\end{align}
where $\gap_{\min}$ is the minimum eigenvalue gap between any two non-degenerate eigenstates, for any $s\in[0,1]$.
We then have
\begin{align}
\|g_q(x)\|&\le \frac{2\|\dot H\|^2}{\gap_{\min}^3T}\int_0^x\left\|\sum_{\nu_{q-2},\nu_{q-3},\ldots} \Ffunc_{\nu_{q-2}}(s_1)\ket{\nu_{q-2}(s_1)} \right\|\mathrm{d}s_1.\nonumber\\
&\le\frac{2\|\dot H\|^2}{\gap_{\min}^3T}\int_0^x\left(\|g_{q-2}(s_1)\|+\|h_{q-2}(s_1)\|\right)\mathrm{d}s_1,\label{eq:gnorm}
\end{align}
where the last inequality follows from the triangle inequality.
Using similar reasoning and our expression for the derivative of $\beta$ in~\eqrefb{eq:betaderiv}, we find that
\begin{align}
\left\|h_q(x)\right\|&\le \left(\frac{\|\dot H\|\|\ddot H\|}{\gap_{\min}^3T}+\frac{6\|\dot H\|^3}{\gap_{\min}^4T}\right)\int_0^x\int_0^{s_1}\left\|\sum_{\nu_{q-2},\nu_{q-3},\ldots} \Ffunc_{\nu_{q-2}}(s_2)\ket{\nu_{q-2}(s_2)} \right\|\mathrm{d}^2s.\nonumber\\
&\le \left(\frac{\|\dot H\|\|\ddot H\|}{\gap_{\min}^3T}+\frac{6\|\dot H\|^3}{\gap_{\min}^4T}\right)\int_0^x\int_0^{s_1}\left(\|g_{q-2}(s_1)\|+\|h_{q-2}(s_1)\|\right) \mathrm{d}^2s\label{eq:hnorm}
\end{align}

We eliminate $\|g_{q-2}\|$ and $\|h_{q-2}\|$ in~\eqrefb{eq:gnorm} and~\eqrefb{eq:hnorm} iteratively by using~\eqrefb{eq:gnorm} and~\eqrefb{eq:hnorm} to provide upper bounds for $\|g_{q-2}\|$ and $\|h_{q-2}\|$ in terms of $\|g_{q-4}\|$ and $\|h_{q-4}\|$.  This process is then repeated $\lfloor q/2 \rfloor$ times to eliminate $\|h_i\|$ and $\|g_i\|$ from our bounds for any $i$.  In the case where $q$ is even, no integrals will remain after performing
this substitution $q/2$ times, whereas if $q$ is odd, the remaining terms will have one integral left that has not been approximated using integration by parts.

Let us now consider the case where $q=2m$.  We expand $f_{\nu_{2m}}$, which is the time-integral part of $C_{2m}$, into a sum of $2^m$ terms
by applying the recursion relations in~\eqrefb{eq:hnorm} and~\eqrefb{eq:gnorm}.  We then group these terms by the number
of integrals present in a given term.  It follows from combinatorics that, for any $\ell \in \{0,1,\ldots,m\}$,
there are $\binom{m}{\ell}$ ways to obtain a term with $2m-\ell$ integrals and hence

\begin{align}
&\|h_{2m}(x)\|+\|g_{2m}(x)\|\nn
&\qquad\le\sum_{\ell=0}^{m}\binom{m}{\ell}\int_0^1\cdots\int_0^{s_{2m-\ell}}\left(\frac{2\|\dot H\|^2}{\gap_{\min}^3T} \right)^{m-\ell}\left(\frac{\|\dot H\|\|\ddot H\|}{\gap_{\min}^3T}+\frac{6\|\dot H\|^3}{\gap_{\min}^4T} \right)^{\ell}\mathrm{d}^{2m-\ell}s
\nonumber\\
&\qquad\le\sum_{\ell=0}^{m}\binom{m}{\ell}\frac{1}{(2m-\ell)!}\left(\frac{2\|\dot H\|^2}{\gap_{\min}^3T} \right)^{m-\ell}\left(\frac{\|\dot H\|\|\ddot H\|}{\gap_{\min}^3T}+\frac{6\|\dot H\|^3}{\gap_{\min}^4T} \right)^{\ell}.\nonumber\\
&\qquad\le\frac{1}{m!(\gap_{\min}T)^m}\sum_{\ell=0}^{m}\binom{m}{\ell}\left(\frac{2\|\dot H\|^2}{\gap_{\min}^2} \right)^{m-\ell}\left(\frac{\|\dot H\|\|\ddot H\|}{\gap_{\min}^2}+\frac{6\|\dot H\|^3}{\gap_{\min}^3} \right)^{\ell}.\nonumber\\
&\qquad=\frac{\Gamma^m}{m!(\gap_{\min}T)^m},\label{eq:lasteven}
\end{align}
where $1/(2m-\ell)!$ is the volume of integration and the last line of~\eqrefb{eq:lasteven} follows from the binomial theorem for
\begin{equation}
\Gamma:=\left(\frac{2\|\dot H\|^2+\|\dot H\|\|\ddot H\|}{\gap_{\min}^2}+\frac{6\|\dot H\|^3}{\gap_{\min}^3} \right).\label{eq:gamdef2}
\end{equation}
By summing~\eqrefb{eq:lasteven} over all $m\ge 2$, we find that
\begin{equation}
\sum_{m=2}^{\infty} \|C_{2m}\|\le e^{\Gamma/(\gap_{\min}T)}-\frac{\Gamma}{\gap_{\min}T}-1.\label{eq:totaleven}
\end{equation}

Our next step is to find $\sum_{m=1}^\infty \|C_{2m+1}\|$.  To do so, we use the fact that a path that contains $2m+1$ jumps
is the concatenation of a path that contains $2m$ jumps and a path that contains one jump.
An upper bound on the contribution of non-adiabatic paths with $2m+1$ jumps is found by taking the norm of
the result for the $2m$ jumps in~\eqrefb{eq:lasteven} and multiplying it by an upper bound on the remaining one-jump
path.  Hence an upper bound on the contribution of these paths is
\begin{align}
\|C_{2m+1}\|&\le\frac{\Gamma^{m}}{m!(\gap_{\min}T)^{m}}\max_{s\in[0,1]}\left\|\sum_{\nu_1}\int_0^s\beta_{\nu_1,\gst}(s')e^{-i\int_0^{s'}\gap_{\gst,\nu_1}(\xi)\mathrm{d}\xi T}\ket{\nu_1(s)}\mathrm{d}s'\right\|,\nonumber\\
&\le\frac{\Gamma^{m}}{m!(\gap_{\min}T)^{m}}\max_{s}\Biggr\|\sum_{\nu_1}\ket{\nu_1(s)}\Biggr(\frac{\beta_{\nu_1,\gst}(s')e^{-i\kappa_{\gst,\nu_1}(s')T}}{-i\gap_{\nu_1,\gst}(s')T}\Biggr|_{0}^s\nn
&-\int_0^s \diff{}{s'}\Biggr(\frac{\beta_{\nu_1,\gst}(s')}{-i\gap_{\nu_1,\gst}(s')T} \Biggr)e^{-i\kappa_{\gst,\nu_1}(s')T} \mathrm{d}s'\Biggr)\Biggr\|,\nonumber\\
&\le\frac{\Gamma^{m}}{m!(\gap_{\min}T)^{m}}\frac{\Gamma}{\|\dot H\| T}.\label{eq:oddjump}
\end{align}
By summing this result over $m=1,2,\ldots$ we find that
\begin{equation}
\sum_{m=1}^\infty \|C_{2_m+1}\|\le\frac{\Gamma}{\|\dot H\|T}\left(e^{\Gamma/(\gap_{\min}T)}-1\right).\label{eq:totalodd}
\end{equation}

Adding~\eqrefb{eq:totaleven} to~\eqrefb{eq:totalodd} yields
\begin{equation}
\sum_{j=3}^\infty \|C_j\|\le\left(1+\frac{\Gamma}{\|\dot H\|T}\right)\left(e^{\Gamma/(\gap_{\min}T)}-1\right)-\frac{\Gamma}{\gap_{\min}T}.\label{eq:gambound}
\end{equation}

The proof of the Theorem then follows from the above results.
By definition, the contribution of the one-jump and two-jump paths to the error in the adiabatic approximation are  $\|\sumOneJumpxpr\|$ and $\|\sumTwoJumpxpr\|$ respectively.  Equation~\ref{eq:lemupperbound} then follows from the triangle inequality.
\end{proofof}

Now all that remains in our proof of the adiabatic approximation is to bound the one- and two-jump contributions
to the error in the adiabatic approximation.  The following lemma contains our bounds for $\|C_1\|$ and $\|C_2\|$.

\begin{lemma}\label{lem:intpartsresult}
Let $H:[0,1]\mapsto \mathbb{C}^{N\times N}$ be differentiable three times then,
\begin{align}
\left\|\sumOneJumpxpr-\left.\sum_{\nu}\frac{e^{-i\phi_\nu}\ket{\nu(1)}\beta_{\nu,\gst}(s_1)e^{-i\gfunc{s_1}{\gst,\nu}T}}{-i\gap_{\gst,\nu}(s_1)T}\right|_{0}^{1}\right\|\le \mathbf{R}_0\label{eq:\Ffunc1bound},
\end{align}
where $\phi_\nu=\int_0^1 E_\nu(\xi)\mathrm{d}\xi T$,
\begin{equation}\mathbf{R}_0:=\frac{2\|\ddot H\|+\|\dddot H\|}{\gap_{\min}^3T^2}+\frac{20\|\dot H\|\|\ddot H\|+12\|\dot H\|^2}{\gap_{\min}^4T^2}+\frac{88\|\dot H\|^3}{\gap_{\min}^5T^2},\end{equation}
and
\begin{align}
\|\sumTwoJumpxpr\|\leq & \frac{\|\ddot H\|^2+4\|\dot H\|^2+5\|\dot H\|\|\ddot H\|}{\gap_{\min}^4T^2}+\frac{12\|\ddot H\|\|\dot H\|^2+30\|\dot H\|^3}{\gap_{\min}^5T^2}+\frac{36\|\dot H\|^4}{\gap_{\min}^6T^2},\label{eq:\Ffunc2lbound}
\end{align}
where $\|\dot H\|=\max_{s\in I} \|\dot H(s)\|$, $\|\ddot H\|=\max_{s\in I} \|\ddot H(s)\|$, $\|\dddot H\|=\max_{s\in I} \|\dddot H(s)\|$ and $\gap_{\min}$
is the minimum energy gap between any two non-degenerate energy levels.
\end{lemma}
The proof of Lemma~\ref{lem:intpartsresult} is similar to our proof of Lemma~\ref{lem:jumpbound}.  The main difference between our proofs is that~\eqrefb{eq:intpartrecur} is
not used, and instead we apply integration by parts twice to guarantee that our approximations are correct to $O(1/T^2)$.  We then use the triangle inequality to find upper bounds for the $O(1/T^2)$ terms and then use the triangle inequality to prove our error bounds.  The details of this
proof are given in Appendix~\ref{appendix:proofintparts}.

The proof of Theorem~\ref{thm:errorbounds} follows from the results of Lemmas~\ref{lem:jumpbound} and~\ref{lem:intpartsresult}.
\begin{proofof}{Theorem~\ref{thm:errorbounds}}
The proof immediately follows by substituting the bounds for $\|C_1\|$ and $\|C_2\|$ in Lemma~\ref{lem:intpartsresult} into the result
of Lemma~\ref{lem:jumpbound}.
\end{proofof}
The proof of our zeroth-order adiabatic theorem, given as Theorem~\ref{thm:zeroorder}, follows directly from this result.  Similarly to the proof of Theorem~\ref{thm:errorbounds}, the proof of Theorem~\ref{thm:jumpbound} follows from combining Lemma~\ref{lem:jumpbound} and Lemma~\ref{lem:intpartsresult} and then showing that the terms that are proportional to $T^{-2}$ are $O(\Delta_1/\gap_{\min}T^2)$, where the timescale $\Delta_1$ is given in Definition~\ref{def:timescale}.

\begin{proofof}{Theorem~\ref{thm:jumpbound}}
Our asymptotic bounds are taken in the limit of small $\Delta_1/[\gap_{\min}T^2]$ and large $\Delta_1$.  From the assumption that $\Delta_1$ is asymptotically greater than any constant, we find by substitution that each of the five terms that are in~\eqrefb{eq:\Ffunc1bound} are proportional to $T^{-2}$ and their sum is $O(\Delta_1/[\gap_{\min}T^2])$.  This implies that
\begin{equation}
C_1=\left.\sum_{\nu}\frac{e^{-i\int_0^1 E_\nu(s)\mathrm{d}s T}\ket{\nu(1)}\beta_{\nu,\gst}(s_1)e^{-i\gfunc{s_1}{\gst,\nu}T}}{-i\gap_{\gst,\nu}(s_1)T}\right|_{0}^{1}+O(\Delta_1/[\gap_{\min}T^2]).
\end{equation}
It follows from substitution into~\eqrefb{eq:\Ffunc2lbound} that $C_2\in O(\Delta_1/[\gap_{\min}T^2])$.

Similarly, we have that
\begin{equation}
\Gamma^2\in O\left(\frac{\|\dot H\|^2\|\ddot H\|^2}{\gap_{\min}^4}+\frac{\|\dot H\|^4\|\ddot H\|}{\gap_{\min}^5}+\frac{\|\dot H\|^6}{\gap_{\min}^6} \right).\label{eq:gamsquared}
\end{equation}
The second term in the above equation is asymptotically greater than or equal to the first term if $\|\ddot H\|\in O(\|\dot H\|^2/\gap_{\min})$.  However, this condition implies that the third term in~\eqrefb{eq:gamsquared} is asymptotically dominant.  Therefore,
\begin{equation}
\Gamma^2\in O\left(\frac{\|\dot H\|^2\|\ddot H\|^2}{\gap_{\min}^4}+\frac{\|\dot H\|^6}{\gap_{\min}^6} \right).\label{eq:gamsquared2}
\end{equation}
The first term in the above equation is dominant only if $\|\dot H\|^2/\gap_{\min}\in O(\|\ddot H\|)$.  It follows from~\eqrefb{eq:gamsquared2} that
\begin{equation}
\Gamma^2\in O\left(\frac{\|\ddot H\|^3}{\gap_{\min}^3}+\frac{\|\dot H\|^6}{\gap_{\min}^6} \right)\in O(\gap_{\min}\Delta_1).\label{eq:gamsquared3}
\end{equation}
From our assumption that $\Delta_1/(\gap_{\min} T^2)\in o(1)$ we find,
\begin{equation}
\frac{\Gamma^2}{\gap_{\min}^2T^2}\in o(1),
\end{equation}
which implies that
\begin{equation}
\frac{\Gamma}{\gap_{\min}T}\in o(1).\label{eq:gamsquared4}
\end{equation}
This result is important because it allows us to replace $\exp(\Gamma/[\gap_{\min}T])$ with a truncated Taylor series, without introducing any error terms that are exponentially large.
Following this strategy, we find from using Taylor's theorem and~\eqrefb{eq:gamsquared4} that the result in~\eqrefb{eq:gambound} is $O(\Delta_1/[\gap_{\min}T^2])$, proving that the contribution of $\|C_3\|,\|C_4\|,\ldots$ is $O(\Delta_1/[\gap_{\min}T^2])$, which proves the theorem.
\end{proofof}
We then can state the proof of our zeroth-order adiabatic theorem as follows.
\begin{proofof}{Theorem~\ref{thm:zeroorder}}
The proof directly follows by applying the triangle inequality to the result of Theorem~\ref{thm:jumpbound} and maximizing each term that is present.
\end{proofof}

Our upper bounds can be further tightened if we assume that the Hamiltonian has only two distinct energy levels for all $s$.
This is because, in this case, all paths that contain an even number of jumps and start in the state $\ket{\gst(0)}$, must end in the state $\ket{\gst(1)}$ because
of our choice of eigenvectors in~\eqrefb{eq:betadef}.  Therefore we neglect all paths that contain an even number of jumps.
This result is stated in the following corollary.
\begin{corollary}\label{cor:jumpbound}
Let $H(s)$ be a twice differentiable Hamiltonian with two distinct energy levels $E_\gst(s)$ and $E_1(s)$, where
the eigenvectors corresponding to $E_1(s)$ are $(N-1)$-fold degenerate.  Given these assumptions, the error in the adiabatic approximation is bounded above and below by,
\begin{align}
\left\|(\openone-P_{\gst}(1))U(1,0)\ket{\gst(0)}-\sumOneJumpxpr\right\|\leq \left(\frac{\Gamma}{\|\dot H\|T}\right)(e^{\Gamma/(\gap_{\min}T)}-1),\label{eq:lemupperbound2}
\end{align}
\end{corollary}
\begin{proof}
The proof of this corollary follows from our proof of
the previous theorem.  Although $E_1$ is $(N-1)$-fold degenerate, this system is an effective two-level system because our choice of eigenstates
prevents any transition between two states with the same energy eigenvalue.  Given this choice of eigenbasis, any path that has an even number of jumps is an adiabatic path, which are excluded from the sum.

Equation~\ref{eq:totalodd} then gives an upper bound on contribution to the error due to paths with more than one-jump because all paths with an even number of jumps do not contribute to the error in this case.  Equation~\ref{eq:lemupperbound2} follows from the definition of $\sumOneJumpxpr$ and the triangle
inequality.
\end{proof}


\section{Search Hamiltonian}\label{sec:searchham}
In this section we apply our error bounds in Theorem~\ref{thm:errorbounds} to the case of a Search Hamiltonian, and compare our results
to the JRS bound~\cite{ruskai:adiabatic}.  In addition, we discuss a new method for choosing the integration
time steps geometrically.

The Search Hamiltonian is a Hamiltonian whose ground state at $s=0$ is an initial guess for the solution to a search problem
 and the ground state at $s=1$ is chosen to be the state that is sought by the algorithm~\cite{farhi:adiabatic,roland:localad}.
If the Hamiltonian evolution obeys the adiabatic approximation then the initial guessed
state will be transformed into the final solution state, thereby solving the search problem.

We denote the initial state $\ket{\psi_0}$ and the marked state $\ket{\psi_1}$, and
the goal of the algorithm is to choose a Hamiltonian that generates an evolution that transforms the initial
guessed state into the marked state in the adiabatic limit.
We evolve our system according to a time-dependent Hamiltonian that is a linear interpolation between these two:
$$H(s)=(1-s)H_0+sH_1=\openone-(1-s)\ketbra{\psi_0}{\psi_0}-s\ketbra{\psi_1}{\psi_1}.$$
Although the search Hamiltonian has $N$ eigenstates, only $\ket{1(s')}$ has non-zero overlap with the ground
state $\ket{\gst(s)}$ for any $s$ and $s'$.  Therefore the search Hamiltonian generates an effective two-dimensional evolution.
The eigenvalue gap between for the two relevant states, if $\ket{\psi_0}=\sum_{x=0}^{N-1} \ket{x}/\sqrt{N}$
and $\ket{\psi_1}$ is a state in the computational basis, is~\cite{roland:localad}
\begin{equation}
\gap_{1,\gst}(s)=\sqrt{1-4\left(1-\frac{1}{N} \right)s(1-s)}.\label{eq:searchgap}
\end{equation}

If we define $\theta=\arccos(|\braket{\psi_0}{\psi_1}|)$ then the instantaneous eigenvectors of the Hamiltonian are
\begin{equation}
\ket{\gst(s)}=\frac{\sin(\theta-\varphi(s))}{\sin\theta}\ket{\psi_0}+\frac{\sin\varphi(s)}{\sin\theta}\ket{\psi_1},
\end{equation}
and
\begin{equation}
\ket{1(s)}=\frac{\cos(\theta-\varphi(s))}{\sin\theta}\ket{\psi_0}
-\frac{\cos\varphi(s)}{\sin\theta}\ket{\psi_1},
\end{equation}
where
\begin{equation}
\varphi(s)=\frac{1}{2}\arctan\left(\frac{s\sin(2\theta)}{(1-s+s\cos(2\theta))}\right).\label{eq:phis}
\end{equation}
These expressions for the eigenvectors clearly show that the eigenvectors undergo rotation in the plane formed
by the span of $\ket{\psi_0}$ and $\ket{\psi_1}$, where the rotation angle is given by $\varphi(s)$.  Similarly,
the rotation speed is given by $\dot\varphi(s)$.
This implies that the eigenstates rotate more slowly near $s=1/2$ than near $s=0$ or $s=1$, it makes sense to use an integration scheme
that takes constant sized steps in $\varphi$ rather than constant sized steps in $s$.  This is because constant
sized steps in $s$ will lead to overly conservative time steps being taken near $s=0$ or $s=1$.

\begin{figure}[t]
\centering
\includegraphics[scale=1]{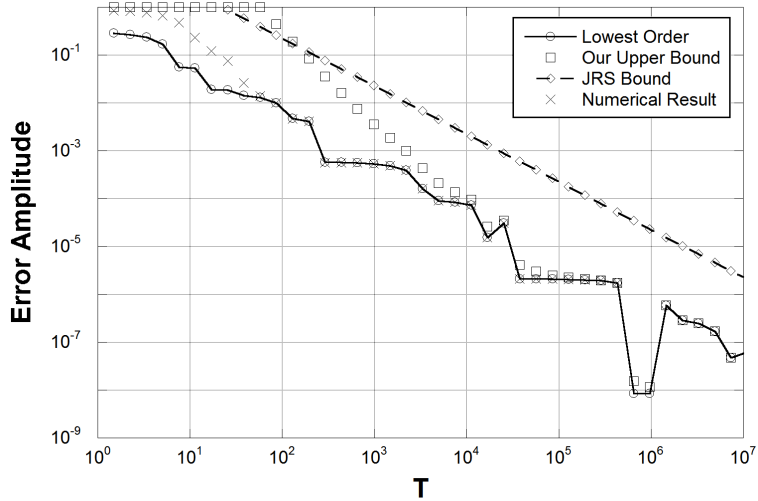}
\caption[Error in the adiabatic approximation for the search Hamiltonian.]{This figure shows the excitation amplitude for a search Hamiltonian with $4$ potential
marked states with $\ket{\psi_0}=\sum_{v=0}^3\ket{v}/\sqrt{4}$.  The jaggedness of the figure appears because of rapid oscillations that occur in the limit of large $T$. \label{fig:adgraph}
}
\end{figure}


%

We find for the case where $\theta=\arccos(1/\sqrt{N})$ that
\begin{equation}
\varphi(s)=\frac{\left[\arctan(\sqrt{N-1}(2s-1))+\arctan(\sqrt{N-1})\right]}{2}.
\end{equation}
Interestingly, this parametrization has been used in the literature to interpolate adiabatic evolution in order to minimize the approximation error~\cite{roland:localad,vandam:adiabaticpower}.  To clarify, we use this interpolation for the timesteps taken in our numerical method, but use linear interpolation for the Hamiltonian.

We use this result to test our upper and lower bounds for the error in the adiabatic approximation by computing the
time-evolution numerically using a modified version of~\eqrefb{eq:productsum} that uses evenly spaced steps in $\varphi$ rather than
$s$.  The results were found by adaptively choosing $L$ such that doubling its value resulted in less than a one percent
change in the computed error in the adiabatic approximation.  To test the results, we also compared the results of the simulation
for small values of $T$ to those found by solving the Schr\"odinger equation using an adaptive quadrature Runge-Kutta integrator using double precision data and found excellent agreement.
However, we found that the Runge-Kutta integrator was unable to provide sufficient accuracy for simulations with large $T$ and
hence we used~\eqrefb{eq:productsum} in its place.

In Fig.~\ref{fig:adgraph} we use these expressions for the eigenvalue gap as well as the eigenvectors to evaluate our upper
bound in Corollary~\ref{cor:jumpbound}, and compare them against the JRS bound~\cite{ruskai:adiabatic} as well
as our numerical data and our lowest order approximation to the error as given by Theorem~\ref{thm:jumpbound}.
In all cases, our upper bound on the error is superior to the JRS bound, in the limit of large $T$.  We also note that the JRS bound
does not converge to the observed error in the limit of large $T$, whereas ours does.  As expected from the fact that our bounds are higher-order than the JRS bound, our bounds are inferior for small values of $T$.

Fig.~\ref{fig:adiabaticlb} demonstrates that our upper bound and lower bounds approach the result of Corollary~\ref{cor:jumpbound} for large $T$.  We
do not plot our numerical data here because for $T>100$ the numerical results are graphically indistinguishable from the result of Corollary~\ref{cor:jumpbound}
and the computation of the error in the approximation becomes prohibitively expensive for $T>10^5$.

We cannot expect our bounds to be asymptotically tight for the search Hamiltonian because the leading order term, as predicted by Theorem~\ref{thm:jumpbound} and Corollary~\ref{cor:jumpbound}, vanishes if
\begin{equation}
T=\frac{2n\pi}{\int_0^1 \gap_{1,\gst}(s)\mathrm{d}s},\label{eq:generalcancelT}
\end{equation}
for any integer $n$.
There are three steps in verifying this claim.  First, we use~\eqrefb{eq:searchgap} to see that $\gap_{\gst,1}(0)=\gap_{\gst,1}(1)$ for the search Hamiltonian.  Second, we note that $\braket{\dot G(1)}{1(1)}=\braket{\dot G(0)}{1(0)}$ for the search Hamiltonian.  Third and last, we substitute these results into~\eqrefb{eq:thmupperbound} to find that the leading order term in both the upper and lower bounds vanish if $T$ is chosen as in~\eqrefb{eq:generalcancelT}.  As we only retained the exact expression for the $O(1/T)$ terms in the error in the adiabatic approximation, the higher order terms will not generically be correct.  Therefore, we cannot expect our error bounds to be asymptotically tight because of the existence of such times.

On the other hand, if $T$ is \emph{not} chosen as in~\eqrefb{eq:generalcancelT} then we can see that the bounds inferred from~\eqrefb{eq:thmupperbound} approach each other as $T$ increases.  This follows because the norms of the derivatives of $H(s)$ are independent of $T$ for the search Hamiltonian, and therefore the difference between the two bounds is $O(1/T^2)$.  As the upper and lower bounds on the error approach each other, we know from the squeeze limit theorem that~\eqrefb{eq:thmupperbound} and~\eqrefb{eq:thmupperbound} approach the correct expression for the error at all points other than those in~\eqrefb{eq:generalcancelT}.
Thus, our error bounds in Theorem~\ref{thm:errorbounds} are not asymptotically tight in the limit of large $T$ for the Search Hamiltonian; however, they are asymptotically tight
for any Hamiltonian that is not explicitly a function of $T$, and for which the leading order contribution does not vanish at some $T \ge 0$.

\begin{figure}[t!]
\centering
\includegraphics[scale=1]{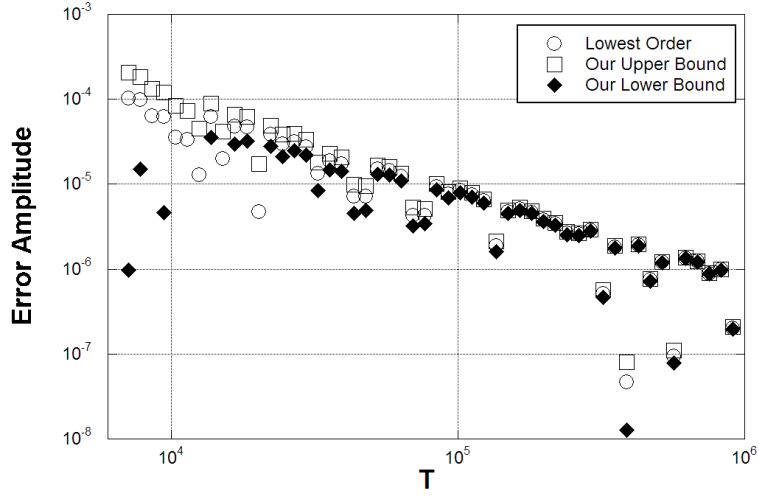}
\caption[Convergence of bounds for error in adiabatic approximation.]{This figure shows that our upper and lower bounds given by Theorem~\ref{thm:errorbounds} for the error in the adiabatic approximation converge for
a search Hamiltonian with $4$ potential marked states, as $T\rightarrow \infty$.}
\label{fig:adiabaticlb}
\end{figure}

\section{Marzlin--Sanders Counterexample}\label{sec:marzlinsanders}
The Marzlin--Sanders counterexample~\cite{marzlin:counter} is often placed at the center of the ongoing debate regarding the conditions that are needed to guarantee the validity of the adiabatic approximation~\cite{mackenzie:adiabaticvalidity,lidar:adiabaticvalidity}.  We examine two issues surrounding the counterexample, in this section. First, we look at the Marzlin--Sanders counterexample to benchmark our criteria for the validity of the adiabatic approximation.  We do this
to show that our sufficiency conditions prohibit the use of the adiabatic approximation for their Hamiltonian.  Second, we examine the possibility that Hamiltonians that are related to the Marzlin--Sanders counterexample may actually behave adiabatically.  In doing so, we provide proofs of Corollaries~\ref{cor:marzsand2} and~\ref{cor:marzsand}.

The counterexample Hamiltonian is of the form~\cite{marzlin:counter},
\begin{equation}
H_{MS}(s)=\left[\omega_0\left(\begin{array}{c}\cos(2\pi s)\\\sin(2\pi s)\\ 0 \end{array}\right)+\frac{2\pi\sin(\omega_0sT)}{T}\left(\begin{array}{c}-\sin(2\pi s)\cos(\omega_0 sT)\\\cos(2\pi s)\cos(\omega_0 sT)\\ \sin(\omega_0 sT)\end{array} \right)\right]\cdot [\sigma_x,\sigma_y,\sigma_z],
\end{equation}
where $\sigma_x$, $\sigma_y$ and $\sigma_z$ are the Pauli operators.
The eigenvalue gap for their Hamiltonian is
\begin{equation}
\gap(s)=\sqrt{\omega_0^2+\left(2\pi\sin(\omega_0 sT)/T\right)^2}.
\end{equation}
The three most important quantities for determining the error using~\eqrefb{eq:lemupperbound} are:
\begin{align}
\|{\partial_s H_{MS}(s)}\|&\in O(1)\label{eq:onederivMS},\\
\|{\partial_s^2 H_{MS}(s)}\|&\in O(T)\label{eq:twoderivMS},\\
\gap_{\min}&\in O(1).
\end{align}
By substituting these scalings into our upper bound in Theorem~\ref{thm:errorbounds}, we expect the error in the
adiabatic approximation to be $\Theta(1)$ and not $o(1)$.  This means that our results suggest, but do not imply, that the Marzlin--Sanders counterexample
does not obey the adiabatic approximation.
On the other hand, if the second derivative of the Marzlin--Sanders Hamiltonian were $o(T)$ instead of $\Omega(T)$, then the adiabatic
approximation would hold for a sufficiently large value of $T$.  We noted this fact in Corollary~\ref{cor:marzsand}, which we prove below.
\begin{proofof}{Corollary~\ref{cor:marzsand}}
Since $\|\dot H\|\in O(1)$ and $\|\ddot H\|\in o(T)$ it follows that $\Gamma\in o(T)$ and $\Gamma/T\in o(1)$.
 We have that \begin{equation}\left(1+\frac{\Gamma}{\|\dot H\|T}\right)(e^{\Gamma/(\gap_{\min}T)}-1)-\frac{\Gamma}{\gap_{\min}T}\in o(1).\end{equation} Therefore Lemma~\ref{lem:jumpbound} implies that
\begin{equation}
\|(\openone-P_{\gst}(1))U(1,0)\ket{\gst}\|\le \sumOneJumpxpr+\sumTwoJumpxpr+ o(1).
\end{equation}
By applying similar reasoning, we find from~\eqrefb{eq:\Ffunc2lbound} that $\|\sumTwoJumpxpr\|\in o(1)$.
As we made no assumptions about $\|\dddot H\|$ in the corollary, we cannot use~\eqrefb{eq:\Ffunc1bound} here.  Instead, we see from~\eqrefb{eq:intparts1j} and Lemma~\ref{lem:betaderiv} that $\|C_1\|\in o(1)$, given that $\|\ddot H\|\in o(T)$.
Therefore if $\|\ddot H(s)\|\in o(T)$ then
\begin{equation}
\|(\openone-P_{\gst}(1))U(1,0)\ket{\gst}\|\in o(1).
\end{equation}
\end{proofof}
The above argument shows that some quantum systems can obey the adiabatic approximation if the second derivatives of the Hamiltonian diverge sufficiently slowly.  We also might imagine that if the derivatives were to diverge even slower then it could be possible that the standard criterion for the error in the adiabatic approximation  holds for that evolution.  This claim was made in Corollary~\ref{cor:marzsand2}, which we prove below using similar reasoning to that used in the proof of Corollary~\ref{cor:marzsand}.
\begin{proofof}{Corollary~\ref{cor:marzsand2}}
Given that $\|\ddot H(s)\|\in o(\sqrt{T})$,  Lemma~\ref{lem:intpartsresult} implies that $\Gamma/T\in o(1/\sqrt{T})$.  Therefore we have that
\begin{equation}\left(1+\frac{\Gamma}{\|\dot H\|T}\right)(e^{\Gamma/(\gap_{\min}T)}-1)-\frac{\Gamma}{\gap_{\min}T}\in o(1/T).\end{equation}  It follows from~\eqrefb{eq:\Ffunc2lbound}
that $\|\sumTwoJumpxpr\|\in o(1/T)$.  We then substitute $\|\ddot H(s)\|\in o(\sqrt{T})$ and $\|\dddot H(s)\|\in o(T)$ into~\eqrefb{eq:\Ffunc1bound} to find that
\begin{equation}
\|\sumOneJumpxpr\|\in O\left(\frac{\|\dot H(0)\|}{(\min_{\nu}\gap_{\gst,\nu}^2(0))T}+\frac{\|\dot H(1)\|}{(\min_{\nu}\gap_{\gst,\nu}^2(1))T}\right).
\end{equation}
 Hence it follows from Lemma~\ref{lem:jumpbound} that
\begin{align}
\|(\openone-P_{\gst}(1))U(1,0)\ket{\gst}\|&\in O\left(\max_{s=0,1}\left[\frac{\|\dot H(s)\|}{(\min_\nu\gap_{\gst,\nu}(s))^2T}\right]\right),
\end{align}
completing the proof.
\end{proofof}



\section{Conclusions}

In this paper, we have shown that the traditional estimate of the error in the adiabatic approximation is valid
if $H(s)$ is three-times differentiable and the norms of these derivatives are bounded above by a constant for all $T \ge 0$.  We have also provided upper and lower bounds for the error in
the adiabatic approximation that are asymptotically tight for many Hamiltonians.  These bounds are highly relevant to applications in quantum information processing because
they can be used to provide bounds for the resources needed to perform a quantum adiabatic computation using a particular
adiabatic path.

We have also shown that some Hamiltonians that have second derivatives that diverge as $T\rightarrow \infty$
behave adiabatically.  Interestingly, the Marzlin--Sanders counterexample Hamiltonian~\cite{marzlin:counter} marginally
fails the criterion in that if the second derivative of their Hamiltonian diverged slightly slower as $T\rightarrow \infty$, then
their proposed Hamiltonian would behave adiabatically in that limit.  This implies that it is not necessary to forbid Hamiltonians
that have second derivatives that diverge with $T$ from the class of Hamiltonians that generate adiabatic evolution.

Our research also leaves open a number of avenues for further work.  First, we have noted that for the search Hamiltonian, the leading order term in the expansion for the error in the adiabatic approximation vanishes at certain times, causing the error to transition from $O(1/T)$ to $O(1/T^2)$ at those times.  This effect could be used to quadratically suppress the probability of a particular transition for an adiabatic evolution, and could lead to promising techniques for improving adiabatic quantum gates, adiabatic algorithms or quantum control protocols.  This effect and its implications
will be discussed elsewhere.

Second, a generalization
of our bounds for unbounded Hamiltonians that act on an infinite-dimensional Hilbert space, or Hamiltonians for which $\ket{\gst(0)}$ is degenerate, would be an interesting contribution to the understanding of adiabatic phenomena.  Similarly, our error bounds can be improved in situations where there are near degeneracies in the eigenvalues of the Hamiltonians, causing $\gap_{\min}$ to be small.  In such circumstances it may be more appropriate to use integration by parts to approximate all jumps that occur across eigenstates that are separated by a large eigenvalue gap, and use a different approximation to handle transitions between nearly degenerate eigenstates.

Third, our analysis of the search Hamiltonian has shown that the error in the adiabatic approximation appears to fall below $2/3$ well before the regime where the adiabatic approximation holds.  This suggests that bounds for the error in the adiabatic approximation are needed in the short-time regime in order to find the time complexity of adiabatic algorithms~\cite{rezakhani:adiabaticexponential}.  Such analysis has not yet been performed for adiabatic algorithms, implying that this remains an interesting avenue for further research.

\appendix

\section{Bounds for Derivatives of $\mathbf{\beta_{\mu,\nu}}$}

We prove our bounds on the derivatives of $\beta$ that are presented in Lemma~\ref{lem:betaderiv}, in this appendix.  These bounds allow us to estimate the contributions to the error in the adiabatic approximation that are incurred due the variation of the tunneling amplitude per unit time, $\beta$.  Our proof of these bounds is given below.

\begin{proofof}{Lemma~\ref{lem:betaderiv}}\label{appendix:proofofbetaderiv}
For convenience we omit the dependence on $s$ of any function below
in order to simplify our resulting expressions (IE: $\beta_{\mu,\nu}=\beta_{\mu,\nu}(s),\gap_{\mu,\nu}=\gap_{\mu,\nu}(s),\ldots$).
Then using this notation, we find by substituting the definition of $\beta_{\mu,\nu}$ into the LHS of~\eqrefb{eq:betaderiv} that
$\Biggr\|\sum_{\nu_q,\ldots,\nu_1}\chi(\vec{\nu})\ket{\nu_q(s')}\dot\beta_{\nu_q,\nu_{q-1}}(s)\Biggr\|$ can be expressed as
\begin{align}
\Biggr\|\sum_{\nu_q,\ldots,\nu_1}\chi(\vec{\nu})\ket{\nu_q(s')}\Biggr(\frac{\bra{\dot \nu_{q}}\dot H \ket{\nu_{q-1}}}{\gap_{\nu_{q},\nu_{q-1}}}&+\frac{\bra{ \nu_{q}}\ddot H \ket{\nu_{q-1}}}{\gap_{\nu_{q},\nu_{q-1}}}\nonumber\\
&+\frac{\bra{\nu_{q}}\dot H \ket{\dot\nu_{q-1}}}{\gap_{\nu_{q},\nu_{q-1}}}-\frac{\dot\gap_{\nu_{q},\nu_{q-1}}\bra{ \nu_{q}}\dot H \ket{\nu_{q-1}}}{\gap_{\nu_{q},\nu_{q-1}}^2}\Biggr)\Biggr\|.\label{eq:betaprime1}
\end{align}
The derivatives of every eigenstate are removed from~\eqrefb{eq:betaprime1} by employing a trick involving the resolution of unity.
For a given eigenstate $\ket{\nu}=\ket{\nu(s)}$ we have that
\begin{equation}
\bra{\dot\nu}=\sum_{m=0}^{N-1}\braket{\dot\nu}{m}\bra{m}=\sum_{m=0}^{N-1} \bra{m}\beta_{m,\nu}.\label{eq:braderivsimp}
\end{equation}
If $\braket{m}{\nu}=0$, then it follows by differentiating that relationship that $\braket{\dot m}{\nu}=-\braket{m}{\dot\nu}$.
Similarly, if $\nu=m$ then $\braket{\dot m}{\nu}=0=-\braket{m}{\dot\nu}$, by our choice of eigenbasis.
Therefore it follows that
\begin{equation}
\ket{\dot\nu}=\sum_{m=0}^{N-1} \ket{m}\braket{m}{\dot\nu}=-\sum_{m=0}^{N-1} \ket{m}\beta_{m,\nu}.\label{eq:ketderivsimp}
\end{equation}

Then by using this property, ~\eqrefb{eq:ketderivsimp},~\eqrefb{eq:braderivsimp} and~\eqrefb{eq:betadef}, the right hand side of~\eqrefb{eq:betaprime1} becomes
\begin{align}
&\Biggr\|\sum_{\nu_q,\ldots,\nu_1}\chi(\vec{\nu})\ket{\nu_q(s')}\Biggr(\sum_{ m\ne\nu_{q}}\frac{\bra{\nu_{q}}\dot H\ket{m}}{\gap_{\nu_{q},m}}\frac{\bra{m}\dot H \ket{\nu_{q-1}}}{\gap_{\nu_{q},\nu_{q-1}}}+\frac{\bra{ \nu_{q}}\ddot H \ket{\nu_{q-1}}}{\gap_{\nu_{q},\nu_{q-1}}}\Biggr.\nonumber\\&\qquad\Biggr.-\sum_{m\ne\nu_{q-1}}\frac{\bra{\nu_{q}}\dot H\ket{m}}{\gap_{\nu_{q},\nu_{q-1}}}\frac{\bra{m}\dot H \ket{\nu_{q-1}}}{\gap_{m,\nu_{q-1}}}-\frac{\dot\gap_{\nu_{q},\nu_{q-1}}\bra{ \nu_{q}}\dot H \ket{\nu_{q-1}}}{\gap_{\nu_{q},\nu_{q-1}}^2}\Biggr)\Biggr\|.\label{eq:betaprime2}
\end{align}
Eq.~\eqrefb{eq:betaprime2} has now been re-written in a form where no derivatives of $\ket{\nu_{q}}$, $\ket{\nu_{q-1}}$ nor $\ket{m}$ are present.
Furthermore the states $\ket{\nu_{q}}$ and $\ket{m}$ appear as parts of projection operators, which have unit norm.  Then by using this property,
the triangle inequality
and lower bounding the eigenvalue gaps in the denominators by $\gap_{\min}$ we find that~\eqrefb{eq:betaprime2} is bounded above by
\begin{equation}
\left(\frac{2\|\dot H\|^2+\|\dot H\|\max_{\nu_{q},\nu_{q-1}}|\dot\gap_{\nu_{q},\nu_{q-1}}|}{\gap_{\min}^2}+\frac{\|\ddot H\|}{\gap_{\min}}\right)\left\|\sum_{\nu_{q-1},\ldots,\nu_1}\chi(\vec{\nu})\ket{\nu_{q-1}(s)}\right\|.
\end{equation}
Eq.~\eqrefb{eq:betaderiv} then follows by using $\max_{\nu_{q},\nu_{q-1}}|\dot\gap_{\nu_{q},\nu_{q-1}}|\le2\|\dot H\|$.

By taking the derivative of the inside of the norm in~\eqrefb{eq:betaprime2}, and by inserting the resolution
 of unity to eliminate all derivatives of eigenvectors, we find that $$\left\|\sum_{\nu_{q},\nu_{q-1},\ldots,\nu_1}\chi(\vec\nu)\ket{\nu_{q}(s')}\ddot\beta_{\nu_{q},\nu_{q-1}}(s)\right\|$$ can be expressed as
\begin{align}
&\Biggr\|\sum_{\nu_{q},\nu_{q-1},\ldots,\nu_1}\chi(\vec\nu)\ket{\nu_{q}(s')}\Biggr(\frac{\partial}{\partial_s}\Biggr(\sum_{ m\ne\nu_{q}}\frac{\beta_{\nu_{q},m}\beta_{m,\nu_{q-1}}\gap_{\nu_{q},m}}{\gap_{\nu_{q},\nu_{q-1}}}
-\sum_{ m\ne\nu_{q}}\frac{\beta_{\nu_{q},m}\beta_{m,\nu_{q-1}}\gap_{m,\nu_{q-1}}}{\gap_{\nu_{q},\nu_{q-1}}}\Biggr)\nn
&+\sum_{m\ne\nu_{q}}\frac{\beta_{\nu_{q},m}\bra{m}\ddot H \ket{\nu_{q-1}}}{\gap_{\nu_{q},\nu_{q-1}}}-\sum_{m\ne\nu_{q-1}}\frac{\beta_{m,\nu_{q-1}}\bra{\nu_{q}}\ddot H \ket{m}}{\gap_{\nu_{q},\nu_{q-1}}}+\frac{\bra{\nu_{q}}\dddot H\ket{\nu_{q-1}}}{\gap_{\nu_{q},\nu_{q-1}}}-\frac{\dot\gap_{\nu_{q},\nu_{q-1}}\bra{\nu_{q}}\ddot H\ket{\nu_{q-1}}}{\gap_{\nu_{q},\nu_{q-1}}^2}\nonumber\\
&-\frac{\ddot\gap_{\nu_{q},\nu_{q-1}}\beta_{\nu_{q},\nu_{q-1}}+\dot\gap_{\nu_{q},\nu_{q-1}}\dot\beta_{\nu_{q},\nu_{q-1}}}{\gap_{\nu_{q},\nu_{q-1}}}+\frac{\dot\gap_{\nu_{q},\nu_{q-1}}^2\beta_{\nu_{q},\nu_{q-1}}}{\gap_{\nu_{q},\nu_{q-1}}^2}\Biggr)\Biggr\|.\label{eq:betaprime4}
\end{align}
By applying the triangle inequality, using the definition of $\beta$ in~\eqrefb{eq:betadef}, the fact that for any $x$ and $y$: $\gap_{x,y}(s)\ge \gap_{\min}$, $\|\beta_{x,y}(s)\|\le \|\dot H\|/\gap_{\min}^2$, $|\dot \gap_{x,y}(s)|\le2\|\dot H\|$ and the fact that $\|\sum_m \ket{m(s')}\bra{m(s)}\|=1$ for any $s$ and $s'$, we find that~\eqrefb{eq:betaprime4} is bounded above by
\begin{align}
\Biggr\|\sum_{\nu_{q},\nu_{q-1},\ldots,\nu_1}&\chi(\vec\nu)\ket{\nu_{q}(s')}\Biggr(\frac{\partial}{\partial_s}\Biggr(\sum_{ m\ne\nu_{q}}\frac{\beta_{\nu_{q},m}\beta_{m,\nu_{q-1}}\gap_{\nu_{q},m}}{\gap_{\nu_{q},\nu_{q-1}}}-\sum_{ m\ne\nu_{q}}\frac{\beta_{\nu_{q},m}\beta_{m,\nu_{q-1}}\gap_{m,\nu_{q-1}}}{\gap_{\nu_{q},\nu_{q-1}}}\Biggr)\Biggr\|\nonumber\\
&+\Biggr\|\sum_{\nu_{q},\nu_{q-1},\ldots,\nu_1}\chi(\vec\nu)\ket{\nu_{q}(s')}\left(\frac{\ddot\gap_{\nu_{q},\nu_{q-1}}\beta_{\nu_{q},\nu_{q-1}}+\dot\gap_{\nu_{q},\nu_{q-1}}\dot\beta_{\nu_{q},\nu_{q-1}}}{\gap_{\nu_{q},\nu_{q-1}}}\right)\Biggr\|
+\frac{4\|\dot H\|\|\ddot H\|}{\gap_{\min}^2}\nn
&+\frac{\|\dddot H\|}{\gap_{\min}}+\frac{4\|\dot H\|^3}{\gap_{\min}}.\label{eq:betaprime5}
\end{align}

Our final result in~\eqrefb{eq:betaderiv2} is found by applying the triangle inequality to this, using~\eqrefb{eq:betaderiv}
and the following result from~\cite{ambainis:adiabatic},
\begin{align}
\ddot \gap_{\nu_{q},\nu_{q-1}} &\leq 2\|\ddot H\| + \frac{8\|\dot H\|^2}{\gap_{\min}}\label{eq:Enu2}.
\end{align}
We then find~\eqrefb{eq:betaderiv2} by using the triangle inequality and equations \eqrefb{eq:Enu2} and~\eqrefb{eq:betaderiv} to upper bound~\eqrefb{eq:betaprime4}.
\end{proofof}

\section{Proof of Lemma \ref{lem:intpartsresult}}\label{appendix:proofintparts}
We present a proof of the upper bounds that are presented in Lemma~\ref{lem:intpartsresult} for the
contribution to the error that results from the sum over all one- and two-jump paths.  Our proof technique
is similar to that used in the proof of Lemma~\ref{lem:jumpbound}.  The most significant difference is that we have to apply integration by
parts twice to approximate the sums of all one- and two-jump terms, which we denote $C_1$ and $C_2$.
\begin{proofof}{Lemma \ref{lem:intpartsresult}}

The contribution of the sum of all one-jump paths to the error in the adiabatic approximation is
\begin{equation}
C_1=\sum_{\mu=1}^{N-1}\ket{\mu(1)}\int_0^1 \beta_{\mu,\gst}(s)e^{-i\int_0^s \gap_{\gst,\mu}(\xi)\mathrm{d}\xi T}\mathrm{d}s.\label{eq:c1xpr}
\end{equation}
We find the following upper bound on $\|C_1\|$
by applying integration by parts twice and the triangle inequality on~\eqrefb{eq:c1xpr}
\begin{align}
\|C_1\|\le&\left\|\left. \sum_{\nu_1} \ket{{\nu_1}(1)} \frac{\beta_{{\nu_1},\gst}(s)e^{-i\kappa_{\gst,{\nu_1}}(s)T}}{-i\gap_{\gst,{\nu_1}}(s)T}\right|_0^1\right\|\nn
&\qquad+\left\|\left.\sum_{\nu_1}\ket{{\nu_1}(1)}  \frac{1}{-i\gap_{\gst,{\nu_1}}(s)T}\left(\frac{\partial}{\partial_s}\frac{\beta_{{\nu_1},\gst}(s)}{-i\gap_{\gst,{\nu_1}}(s)T}\right)e^{-i\kappa_{\gst,{\nu_1}}(s)T}\right|_0^1\right\|\nonumber\\
&\qquad+\left\|\int_0^1 \sum_{{\nu_1}}\ket{{\nu_1}(1)} \left(\frac{\partial}{\partial_s}\frac{1}{-i\gap_{\gst,{\nu_1}}(s)T}\left(\frac{\partial}{\partial_s}\frac{\beta_{{\nu_1},\gst}(s)}{-i\gap_{\gst,{\nu_1}}(s)T}\right)\right)e^{-i\kappa_{\gst,{\nu_1}}(s)T}\mathrm{d}s \right\|.
\end{align}
Using our bounds for the derivatives of $\beta$ in~\eqrefb{eq:betaderiv} and~\eqrefb{eq:betaderiv2}, the bound $|\dot \gap_{{\nu_1},\nu}|\le 2\|\dot H\|$ and the upper bound for $|\ddot\gap|$ in~\eqrefb{eq:Enu2} we find that

\begin{equation}
\left\|C_1\right\|\le \left\|\left. \sum_{\nu_1} \ket{{\nu_1}(1)} \frac{\beta_{{\nu_1},\gst}(s)e^{-i\kappa_{\gst,{\nu_1}}(s)T}}{-i\gap_{\gst,{\nu_1}}(s)T}\right|_0^1\right\|+\mathbf{R}_0.\label{eq:\Ffunc1bound2}
\end{equation}
where
\begin{equation}
\mathbf{R}_0:=\frac{2\|\ddot H\|+\|\dddot H\|}{\gap_{\min}^3T^2}+\frac{20\|\dot H\|\|\ddot H\|+12\|\dot H\|^2}{\gap_{\min}^4T^2}+\frac{88\|\dot H\|^3}{\gap_{\min}^5T^2}.
\end{equation}

Now we will apply similar reasoning to approximate $C_2$, which is the contribution of all non-adiabatic two-jump paths to the error in the adiabatic approximation.  This error can be expressed as,
\begin{equation}
\sum_{{\nu_2}\ne 0}\sum_{\nu_1}\ket{{\nu_2}(1)}\int_0^1\beta_{{\nu_2},\nu_1}(s_1)e^{-i\int_0^{s_1}\gap_{\nu_1,{\nu_2}}(\xi)\mathrm{d}\xi T}\int_0^{s_1} \beta_{\nu_1,\gst}(s_2)e^{-i\int_0^{s_2}\gap_{\gst,\nu_1}(\xi)\mathrm{d}\xi T}\mathrm{d}s_2 \mathrm{d}s_1.
\end{equation}
We use integration by parts once on the outer-most integral and the fundamental theorem of calculus to find
\begin{align}
&\sum_{{\nu_2}\ne 0}\sum_{\nu_1}\ket{{\nu_2}(1)}\Biggr( \frac{\beta_{{\nu_2},\nu_1}(1)e^{-i\int_0^1\gap_{\nu_1,{\nu_2}}(\xi)\mathrm{d}\xi T}}{-i\gap_{\nu_1,{\nu_2}}(1)T}\int_0^1\beta_{\nu_1,\gst}(s_2)e^{-i\int_0^{s_2} \gap_{\gst,\nu_1}(\xi)\mathrm{d}\xi T}\mathrm{d}s_2\nonumber\\
&\qquad-\int_0^1\left(\diff{}{s_1}\frac{\beta_{{\nu_2},\nu_1}(s_1)}{-i\gap_{\nu_1,{\nu_2}}(s_1)T}\right)e^{-i\int_0^{s_1}\gap_{\nu_1,{\nu_2}}(\xi)\mathrm{d}\xi T}\int_0^{s_1}\beta_{\nu_1,\gst}(s_2)e^{-i\int_0^{s_2}\gap_{\gst,\nu_1}(\xi)\mathrm{d}\xi T}\mathrm{d}s_2\mathrm{d}s_1\nonumber\\
&\qquad-\int_0^1\frac{\beta_{{\nu_2},\nu_1}(s_1)\beta_{\nu_1,\gst}(s_1)e^{-i\int_0^{s_1}\gap_{\gst,\nu_1}(\xi)+\gap_{\nu_1,{\nu_2}}(\xi)\mathrm{d}\xi T}}{-i\gap_{\nu_1,{\nu_2}}(s_1)T}\mathrm{d}s_1\Biggr).~\label{eq:2jumpintparts1b}
\end{align}
 We must have that $E_{{\nu_2}}(s)\ne E_{\gst}(s)$ to use integration by parts again on this expression, which is ensured by our assumption that the initial state is non-degenerate.  Using integration
by parts again to estimate the inner-most integrals in~\eqrefb{eq:2jumpintparts1b} we have that
\begin{align}
\sum_{{\nu_2},\nu}\ket{{\nu_2}(1)}&\Biggr[ \frac{\beta_{{\nu_2},\nu_1}(1)e^{-i\int_0^1\gap_{\nu_1,{\nu_2}}(\xi)\mathrm{d}\xi T}}{-i\gap_{\nu_1,{\nu_2}}(1)T}\Biggr (\frac{\beta_{\nu_1,\gst}(s_2)e^{-i\int_0^{s_2}\gap_{\gst,\nu_1}(\xi)\mathrm{d}\xi T}}{-i\gap_{\gst,\nu_1}(s_2)T}\Biggr|_0^{1}\nn
&\qquad-\int_0^1\Biggr (\diff{}{s_2}\frac{\beta_{\nu_1,\gst}(s_2)}{-i\gap_{\gst,\nu_1}(s_2)T} \Biggr )e^{-i\int_0^{s_2}\gap_{\gst,\nu_1}(\xi)\mathrm{d}\xi T}\mathrm{d}s_2 \Biggr )\nonumber\\
&-\int_0^1 \Biggr (\diff{}{s_1}\frac{\beta_{{\nu_2},\nu_1}(s_1)}{-i\gap_{\nu_1,{\nu_2}}(s_1)T} \Biggr )e^{-i\int_0^{s_1}\gap_{\nu_1,{\nu_2}}(\xi)\mathrm{d}\xi T}\Biggr\{\frac{\beta_{\nu_1,\gst}(s_2)e^{-i\int_0^{s_2}\gap_{\gst,\nu_1}(\xi)\mathrm{d}\xi T}}{-i\gap_{\gst,\nu_1}(s_2)T}\Biggr|_0^{s_1}&\nonumber\\
&\qquad-\int_0^{s_1}\Biggr (\diff{}{s_2}\frac{\beta_{\nu_1,\gst}(s_2)}{-i\gap_{\gst,\nu_1}(s_2)T} \Biggr )e^{-i\int_0^{s_2}\gap_{\gst,\nu_1}(\xi)\mathrm{d}\xi T}\mathrm{d}s_2\Biggr\}\mathrm{d}s_1\nn
&-\frac{\beta_{{\nu_2},\nu_1}(s_1)\beta_{\nu_1,\gst}(s_1)e^{-i\int_0^{s_1}\gap_{\gst,\nu_1}(\xi)+\gap_{\nu_1,{\nu_2}}\mathrm{d}\xi T}}{-\gap_{\nu_1,{\nu_2}}(s_1)(\gap_{\gst,\nu_1}(s_1)+\gap_{\nu_1,{\nu_2}}(s_1))T}\Biggr|_0^1\nonumber\\
&-\int_0^1\Biggr ( \diff{}{s_1}\frac{\beta_{{\nu_2},\nu_1}(s_1)\beta_{\nu_1,\gst}(s_1)}{-\gap_{\nu_1,{\nu_2}}(s_1)(\gap_{\gst,\nu_1}(s_1)+\gap_{\nu_1,{\nu_2}}(s_1))T^2} \Biggr )e^{-i\int_0^{s_1}\gap_{\gst,\nu_1}(\xi)+\gap_{\nu_1,{\nu_2}}\mathrm{d}\xi T}\mathrm{d}s_1 \Biggr].\label{eq:2jumpintparts2}
\end{align}
Then using the triangle inequality, our upper bounds for the derivatives of $\beta$ in~\eqrefb{eq:betaderiv} and~\eqrefb{eq:betaderiv2}, the bound $|\dot \gap_{{\nu_2},\nu}|\le 2\|\dot H\|$ and the upper bound for $|\ddot\gap|$ in~\eqrefb{eq:Enu2} we find by substitution that
\begin{equation}
\|C_2\|\le \frac{\|\ddot H\|^2+4\|\dot H\|^2+5\|\dot H\|\|\ddot H\|}{\gap_{\min}^4T^2}+\frac{12\|\ddot H\|\|\dot H\|^2+30\|\dot H\|^3}{\gap_{\min}^5T^2}+\frac{36\|\dot H\|^4}{\gap_{\min}^6T^2}.
\end{equation}
\end{proofof}

\acknowledgements
We thank C. Dohotaru, J. Rashid, A. Rezakhani, and B.~C. Sanders for many helpful comments.  This work was supported by Canada' s Natural Sciences and Engineering Research Council (NSERC), the Canadian Institute for Advanced Research (CIFAR), Mathematics of Information Technology and Complex Systems (MITACS)'s Canadian Network of Centres of Excellence (NCE) program,  the United States Army Research Office (USARO), and QuantumWorks.


\end{document}